\documentclass{llncs}

\usepackage[utf8]{inputenc}
\usepackage{bussproofs}
\usepackage{amsmath}
\usepackage{amssymb}
\usepackage[disable]{todonotes}
\usepackage{stmaryrd}
\definecolor{linkblue}{rgb}{0.1,0,0.5}
\usepackage[pdfborder={0 0 0},colorlinks=true,linkcolor=linkblue,citecolor=linkblue,filecolor=linkblue,urlcolor=linkblue]{hyperref}
\usepackage{scrextend}
\usepackage{graphicx}

\def\extended{EXTENDED}

 \def\version{EXTENDED}

\def\true{TRUE}

\def\submission{FALSE}

\def\true{TRUE}

\def\bottom{\bot}

\def\extended{EXTENDED}

\newcommand*\externallink[1]{\href{{#1}}{\includegraphics[scale=0.7]{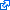}}}

\newcommand{\step}{\text{ }\rightarrow\text{ }}
\newcommand{\mstep}{\text{ }\stackrel{*}{\rightarrow}\text{ }}

\newenvironment{corners}{$\operatornamewithlimits{}_\llcorner^{\ulcorner}$}{$\operatornamewithlimits{}_\lrcorner^{\urcorner}$}

\newcommand{\mtag}[1]{\mathbf{#1}}

\newcommand{\mtagsmall}[1]{\tiny{\mathbf{#1}}}

\newcommand{\mtagtypesafe}{\mtag{typesafe}} 
\newcommand{\mtagtypesafesmall}{\mtagsmall{typesafe}} 
\newcommand{\mtagfailsafe}{\mtag{failsafe}} 
\newcommand{\mtagterminates}{\mtag{terminates}} 
\newcommand{\mtagterminatessmall}{\mtagsmall{terminates}} 

\newcommand{\mtags}{\mathcal{T}ags} 

\newcommand{\uninitialized}{\boxtimes} % alt: \dagger\nabla\mho\textinterrobangdown\textreferencemark\XBox\maltese\cyrf\cyrzh

\newcommand{\mtagvarX}{\mathbf{tags}}
\newcommand{\mtagvarXsmall}{\tiny{\mathbf{tags}}}
\newcommand{\mtagvarY}{\mathbf{tags'}}
\newcommand{\mtagtag}{\mathbf{tag}}
\newcommand{\mtagtagsmall}{\tiny{\mathbf{tag}}}
\newcommand{\mmybox}[1]{\tikzpicture[baseline=(text.base)]{\node [draw,rounded corners,fill=red!20] (text) {hihi};}}

\newcommand{\oC}{\theta}
\newcommand{\creation}{\gamma}

\newcommand{\Var}{\mathfrak{V}}
\newcommand{\Class}{\mathfrak{C}}
\newcommand{\Method}{\mathfrak{M}}

\newcommand{\TypeAL}{\mathbb{T}}
\newcommand{\tags}{\mathcal{T}ags}

\newcommand{\LVar}{\Var_L}
\newcommand{\IVar}{\Var_I}
\newcommand{\SVar}{\Var_S}

\newcommand{\result}{\mathbf{r}}

\newcommand{\typeALobj}{\mathbb{O}}
\newcommand{\typeALbool}{\mathbb{B}}
\newcommand{\typeALnum}{\mathbb{N}}

\newcommand{\typeALstring}{\mathbb{S}}

\newcommand{\free}{\mathit{free}}

\newcommand{\Vnull}{\mathit{null}}
\newcommand{\this}{\PL{self}}
\newcommand{\self}{\this}

\newcommand{\PL}[1]{\text{#1}}

\newcommand{\Vector}[1]{\overrightarrow{#1}}
\newcommand{\emptyP}{\result}

\newcommand{\pure}{\varepsilon}

\newcommand{\BaseTypeX}{\mathbb{T}}

\newcommand{\mapbool}{\mathbb{B}}  %TODO: 
\newcommand{\mapnum}{\mathbb{N}}

\newcommand{\semantics}{\mathcal{M}}
\newcommand{\selector}{\mathcal{S}}

\newcommand{\langDyn}{\textbf{dyn}}
\newcommand{\langAL}{\textbf{AL}}

\newcommand{\PSHL}{\mathcal{H}}
\newcommand{\PSAL}{\mathcal{T}}

\newcommand{\info}[1]{{\small #1}}

\begin{document}

\title{A Sound and Complete Hoare Logic for Dynamically-typed, Object-Oriented Programs \ifx \version\extended \newline -- Extended Version -- \fi}

%\titlebanner{\ifx \version\extended extended version -- draft \else preprint \fi}        % These are ignored unless
%\preprintfooter{Practical Verification of Dynamically Typed Programs\ifx \version\extended -- Extended Version -- Draft\fi}   % 'preprint' option specified.

\ifx \submission\true
\author{-- omitted for submission --}
\institute{-- omitted for submission --}
\else
%~(\raisebox{0pt}{\includegraphics[height=.6em]{envelope}})
\author{Bj\"orn Engelmann\thanks{Corresponding author} \and Ernst-R\"udiger Olderog}
\institute{University of Oldenburg, Germany\thanks{This work is supported by the German Research Foundation through the Research Training Group (DFG GRK 1765) SCARE (\href{http://www.scare.uni-oldenburg.de}{www.scare.uni-oldenburg.de}).} \\ 
           \email{\{bjoern.engelmann, ernst.ruediger.olderog\}@informatik.uni-oldenburg.de}}
\fi

\maketitle

\begin{abstract}
A simple dynamically-typed, (purely) object-oriented language is defined. A structural operational
semantics as well as a Hoare-style program logic for reasoning about programs in the language in
multiple notions of correctness are given. The Hoare logic is proved to be both sound and (relative)
complete and is -- to the best of our knowledge -- the first such logic presented for a
dynamically-typed language.
\end{abstract}

%paper outline
\section{Introduction \& Related Work}
\label{sec:introduction}

While dynamic typing itself was introduced with the advent of LISP decades ago and more and
more dynamically-typed programs are written as languages like JavaScript, Ruby and Python are
gaining popularity, to the current day, no sound and complete program logic has been published
for any such language.

In an attempt to bridge this Gap between static- and dynamically-typed languages, we focus our inquiry on
completeness (for closed programs) and on studying the proof-theoretic implications of dynamic typing.
This differentiates our work from other axiomatic semantics published mainly for JavaScript
\cite{QinEtAl2011,Gardner12programlogicJS} as their focus lies more on soundness and direct
applicability to real-world programming languages.
\todo{ cite operational semantics? (FullMonty2013,Smeding2009,Filaretti2014)}
% - JavaScript
%   * Sep.Log from London
%   * Chinese guys
% - Python
%   - K-Python
%   - operational (Smeding + The Full Monty)
% - PHP
%   - K-PHP

% This is not to say that there has not been considerable interest in the axiomatic verification
% of dynamically-typed languages 
% However, all of those deal with the real-world programming language JavaScript\footnote{for Python and PHP
% only operational semantics have been published so far} and
% -- to be fair --
% the authors of both deserve respect for investing the effort of providing at least
% a soundness proof for a logic of this scale.

Hence, to avoid getting tangled in the details of any real-world programming language, we 
introduce a small dynamically-typed object-oriented (OO) language called \langDyn{}\footnote{%
One may ask whether it is at all possible to obtain a sound and relatively complete
Hoare logic for \langDyn{} in light of Clarke's incompleteness result \cite{Clarke1979}.
However, Clarke's argument is not applicable to \langDyn{} for various reasons elaborated in
\ifx \version\extended
appendix~\ref{app:turing_complete}.
\else
\cite[appendix C]{extended_version}.
\fi}.

Additionally, in previous work \cite{EngelmannOlderogFlick2015} the authors developed a technique for
reducing the effort of verifying a dynamically-typed program to the level of verifying an equivalent
statically-typed one.
This technique, however, assumed the existence of a sound and complete program logic for the dynamically-typed
language. The current work hence substantiates this assumption.

% main contributions
% - first sound and complete Hoare Logic for a dynamically-typed language
% - Tagged Hoare Logic
%   - partial correctness
%   - strong (= failsafe) partial correctness
%   - typesafe partial correctness
%   - total correctness
% - Fixing Completeness Issue for Total Correctness

Besides presenting the Hoare logic, there are further technical contributions:

\noindent\textbf{1)} Tagged Hoare Logic, a novel notation for Hoare triples making the notion of correctness explicit
 and thereby allowing the (previously separated) Hoare logics for partial correctness, strong (= failsafe)
 partial correctness, typesafe partial correctness and total correctness to be merged into a single proof system
 and to concisely express the rules of this system.

\noindent\textbf{2)} A novel technique to specify loop variants circumventing a common
 incompleteness issue in Hoare logics for total correctness (see proof of Theorem~\ref{thm:comp_rec_meth}).

 As detailed in
\ifx \version\extended
Section~\ref{sec:translational_approach},
\else
\cite[Section 7]{extended_version},
\fi{}
we consider our results as a stepping stone towards similar proof systems for real-world languages.

 Our paper is oganized as follows.
In Section~\ref{sec:dyn}, we introduce the language \langDyn{}.
In Section~\ref{sec:dyn:operational}, its operational semantics is defined.
In Section~\ref{sec:dyn:axiomatic}, its axiomatic semantics (Hoare logic) is introduced.
In Section~\ref{sec:soundness}, we briefly touch upon soundness of this Hoare logic,
and in Section~\ref{sec:completeness}, we prove its (relative) completeness for closed programs.

\noindent\textbf{Notation:} $\mathbb{N}^n_m \equiv \{n,...,m\}$, $\mathbb{N}_m \equiv \mathbb{N}^0_m$, 
$S_1 S_2$ denotes concatenation of the sequences $S_1$ and $S_2$, $\{S\}$ is the set of all elements of
the sequence $S$. 
% 
% The result of substituting $t$ for $v$ in $p$ is denoted
% $p[v := t]$\ifx \version\extended or $p\frac{v}{t}$ \fi.
 
% - state update, substitution and Cook-style substitution
% - |N_k
% - \exists x:T \bullet / \forall
% - sequences + juxtaposition = concatenation

% - dyn
\section{Dynamically-Typed Programs}
\label{sec:dyn}
%   - design principles

We will study a language called \textbf{dyn}, whose syntax is depicted in Figure~\ref{fig:dyn-syntax}. Like its popular real-world siblings JavaScript,
Ruby and Python, \textbf{dyn} is a dynamically-typed purely OO-programming language. However,
to focus our inquiry on dynamic typing, we chose not to model other features commonly found in
these languages like method update, closures or eval().

As customary in such languages, \langDyn{} desugars operations to method calls.
Consequently, the only built-in operation in \langDyn{} is object equality. Everything else is defined in
\langDyn{} itself. However, a syntactic distinction between built-in operations and method calls is necessary
for the convenient distinction between (side-effect-free) expressions and (side-effecting) statements. In order
to make \langDyn{} programs resemble their real-world counterparts, we had to allow method calls as well as
assignments in expressions.
For example, $ \PL{a} := \PL{b} := 5 $ is a valid \langDyn{} expression with the side-effect of assigning
5 to both $\PL{a}$ and $\PL{b}$.

Since types in \langDyn{} are a property of values rather than variables, there is no need to declare the latter.
Following its real-world counterparts, both local- and instance variables in \langDyn{} are created upon their
first assignment. 
% For type systems are also a handy way of catching typos,
Accessing a variable that has
not been assigned before results in a (runtime) type error.

Other reasons for type errors are non-boolean conditions in conditionals or while-loops
and method call receivers whose class does not support a method matching name and arity of the call
(MethodNotFound).

\todo{Interpretive model for dynamically-typed languages differs from the one used by Cook}

\todo{mapping predicates - in particular those for recursive data types like lists, strings, trees, etc. are usually recursive}

\section{Operational Semantics}
\label{sec:dyn:operational}

In Figure~\ref{fig:dyn_op_sem}, we define an operational semantics of \langDyn{}
in the style of Hennessy and Plotkin~\cite{HP79,Plo04}.
It is based on a set $\mathit{Conf}$ of \emph{configurations}, which are
pairs $C =\langle s, \sigma \rangle$ consisting of a statement $s$ of \langDyn{}
and a state $\sigma$, assigning values to variables.
By syntax-directed rules, the operational semantics defines which
\emph{transitions} $\langle s, \sigma \rangle \step \langle s', \sigma' \rangle$
are possible between configurations.

As \langDyn{} is a purely OO-language, the value domain is the set
$\typeALobj$  of \emph{objects},
including the special objects $\Vnull$ (the usual OO-null value) and
$\uninitialized$ (marking non-existing variables). The definition of states and
state updates is standard and therefor omitted (see e.g.~\cite{AptDeBoerOlderogBook2009}\todo{chapter, page?}).

For a given program, we denote the set of all variables as $\Var = \Var_L \uplus \Var_I \uplus \Var_S$
\info{where $\Var_L$ is the set of local variables, $\Var_I$ the set of instance variables and $\Var_S = \{\self, \result\}$
the set of special variables}. $\self$ is special because it cannot be assigned to in programs and $\result$ will
be explained below. We also use the set of all classes $\Class$ with each class $C \in \Class$ having a set of methods
$\Method_C$ and $\Method = \bigcup_{C \in \Class}\Method_C$.

Usually, in a structural operational semantics, expressions are assumed to be side-effect-free and
the effect of assignments can hence be expressed as an axiom $\langle v := e, \sigma \rangle \step \langle \emptyset, \sigma[v := \sigma(e)] \rangle$.
In \langDyn{}, however, expressions are side-effecting. We hence need
to evaluate the assignment $v := e$ in two steps: first evaluating the expression $e$ and then assigning its
resulting value to the variable $v$. Furthermore, we need an interface between these two steps: A way by
which the assignment can determine the result of the previously evaluated expression $e$.
For this purpose, we introduce a special variable $\result$ of type $\typeALobj$ as
well as the convention that every expression or statement will store its result in $\result$.
Note that this construction works only due to dynamic typing: In a statically-typed programming language,
expressions would evaluate to values of different types which could not well be assigned to a single variable.
The choice of object as the unifying supertype of all values is common in pure OO-languages:
When everything is an object, clearly every expression will evaluate to one.
Furthermore, as $\result$ is the only statement that does not change anything (not even $\result$), we
define the empty program as $\result$, stipulate $(\emptyP; s) \equiv (s; \emptyP) \equiv s$ for all statements $s$
and call the configurations $\langle \emptyP, \sigma \rangle$ for some state $\sigma$ final.

For \langDyn{}, we use class-based OO and model object creation as activation\footnote{Assuming an infinite sequence of already existing, but deactivated objects, object creation boils down to picking the next one and marking it as ``activated''.}.
We introduce a ``representative'' object $\oC_C$ for each class $C$ as well as a special instance variable
$@\textbf{c}$ not allowed to occur in programs for maintaining both the instance-class relation
and the activation state of each object.

We call an object $o$ with $o.@\mathbf{c} = \Vnull$ \emph{inactive}, meaning it is ``not yet created''.
Initially, all objects (except $\Vnull$ and the representatives $\oC_C$ for each class $C$) are inactive.
We suppose an infinite enumeration of objects $o_1,o_2,...$ containing every object
(both active and inactive) exactly once and introduce a function $\creation{}: \Sigma \mapsto \typeALobj$
mapping every state $\sigma \in \Sigma$ to the object $o_k$ with the least index $k$ that is inactive in
$\sigma$.

Upon its creation, an object $o$ is assigned a class $C$ and is henceforth regarded an instance of $C$.
Technically, this is achieved by resetting the value of $o.@\mathbf{c}$ to $\oC_C$ (see the rule for
object creation). We use $\mathit{init}_C$ to denote the initial (internal) state of an object of class $C$:
$\mathit{init}_C.@\mathbf{c} = \oC_C$ and $\mathit{init}_C.@v = \uninitialized$ for all $@v \in \Var_I \setminus \{@\mathbf{c}\}$.

We can then formally define the predicate $bool(o)$ and $bool(o,b)$ used in Figure~\ref{fig:dyn_op_sem} to check for
boolean values as

$bool(o) \equiv o.@\mathbf{c} = \oC_{bool} \text{ for all } o \in \typeALobj \text{ and}$

$bool(o,b) \equiv bool(o) \wedge b \leftrightarrow o.@to\_ref \not= \Vnull \footnote{Other methods to distinguish the values true and false are conceivable.} \text{ for all } o \in \typeALobj, b \in \typeALbool \;.$
\begin{figure}[p]
\begin{minipage}{0.5\textwidth}
% first column
\underline{Syntax of \textbf{dyn}:}

$\mathit{Prog} \ni \pi ::= \Vector{\mathit{class}} \PL{ } s$

$\mathit{Class} \ni \mathit{class} ::= \mathbf{class}\PL{ C $<$ C }\{ \Vector{\mathit{meth}} \}$

$\mathit{Meth} \ni \mathit{meth} ::= \mathbf{method}\PL{ m}(\Vector{\PL{u}}) \{ s \}$

\hspace*{0.8cm}$\mid \mathbf{rename}\PL{ m m}$

$\mathit{Stmt} \ni s ::= s;s \mid e$

$\mathit{Expr} \ni e ::= \PL{null} \mid \PL{u} \mid \PL{@v} \mid \self \mid e == e$

\hspace*{0.8cm}$\mid e \PL{ is\_a? C} \mid e.\PL{m}(\Vector{e}) \mid \mathbf{new}\PL{ C}(\Vector{e})$

\hspace*{0.8cm}$\mid \PL{u} := e \mid \mathbf{if}\PL{ } e \PL{ }\mathbf{then}\PL{ } s \PL{ }\mathbf{else}\PL{ } s \PL{ }\mathbf{end}$

\hspace*{0.8cm}$\mid \PL{@v} := e \mid \mathbf{while}\PL{ } e \PL{ }\mathbf{do}\PL{ } s \PL{ }\mathbf{done}$

($\PL{u} \in \Var_L, \PL{@v} \in \Var_I, \PL{C} \in \Class, \PL{m} \in \Method$)
\end{minipage}
\begin{minipage}{0.5\textwidth}
\underline{Syntactic sugar:} 

$e_1 \oplus e_2 \equiv e_1.\PL{m}_{\oplus}(e_2)$

$\PL{\textbf{if} } e \PL{ \textbf{then} } s \PL{ \textbf{end}} \equiv$

$\PL{\textbf{if} } e \PL{ \textbf{then} } s \PL{ \textbf{else} null \textbf{end}}$

$\mathit{false} \equiv \mathbf{new}\PL{ bool(null)}$ \hfill $[] \equiv \mathbf{new}\PL{ list()}$

$\mathit{true} \equiv \mathit{false}.\PL{not()}$ \hfill $[...,o] \equiv \PL{ [...].add(o)}$

$0 \equiv \mathbf{new}\PL{ num(null)}$

$n \equiv \PL{$(n-1)$.succ()}$ for $n \in \mathbb{N}$

$"" \equiv \mathbf{new}\PL{ string(null, null)}$,

$\PL{"...a"} \equiv \PL{"..."}.\PL{addchar}(n_{\PL{a}})$
where $n_{\PL{a}} \in \mathbb{N}$ is the ASCII-code of character $\PL{a}$.

\end{minipage}
%\begin{minipage}{0.5\textwidth}
%\underline{Basic data types in \textbf{stat}:}

%$\mathit{true}, \mathit{false} \in \mathit{Cnst_s}$

%$\wedge, \vee, \rightarrow: \mathbb{B} \times \mathbb{B} \mapsto \mathbb{B}  \in \mathit{Op_s}$

%$0,1,... \in \mathit{Cnst_s}$

%$+,*,div: \mathbb{N} \times \mathbb{N} \mapsto \mathbb{N} \in \mathit{Op_s}$

%$=: \mathbb{N} \times \mathbb{N} \mapsto \mathbb{B} \in \mathit{Op_s}$

%$"", "a", ... \in \mathit{Cnst_s}$

%$+: \mathbb{S} \times \mathbb{S} \mapsto \mathbb{S} \in \mathit{Op_s}$

%$s[n]: \mathbb{S} \times \mathbb{N} \mapsto \mathbb{N} \in \mathit{Op_s}$

%$|s|: \mathbb{S} \mapsto \mathbb{N} \in \mathit{Op_s}$

%\end{minipage}

\caption{\label{fig:dyn-syntax}Syntax of \textbf{dyn}}
\end{figure}
\todo{$e==e$ and $e \PL{ is\_a? C}$ were added. Also add to semantics, etc.}
\todo{mapping $\oplus$ to $m_\oplus$}

\begin{figure}[p]
\begin{enumerate}
 \item $\langle \Vnull, \sigma\rangle  \step \langle\emptyP, \sigma[\result := \Vnull]\rangle $

 \item $\langle v, \sigma\rangle  \step \langle\emptyP, \sigma[\result := \sigma(v)]\rangle $ \info{where $v \in \Var$ and $\sigma(v) \not= \uninitialized$}

 \item $\langle v, \sigma\rangle  \step \langle \emptyP, \textbf{typeerror} \rangle $ \info{where $v \in \Var$ and $\sigma(v) = \uninitialized$}

 \item \AxiomC{$\langle e, \sigma\rangle  \mstep \langle\emptyP, \tau\rangle $}
       \RightLabel{\info{where $v \in \Var$}}
       \UnaryInfC{$\langle v := e, \sigma\rangle  \step \langle\emptyP, \tau[v := \tau(\result)]\rangle $}
       \DisplayProof
 
 \item \AxiomC{$\langle s_1, \sigma\rangle  \step \langle s_2, \tau\rangle $}
       \UnaryInfC{$\langle s_1; s, \sigma\rangle  \step \langle s_2; s, \tau\rangle $}
       \DisplayProof

 \item \AxiomC{$\langle e, \sigma\rangle  \mstep \langle\emptyP, \tau\rangle $}
       \AxiomC{$\mathit{bool}(\tau(\result), \mathrm{true})$}
       \BinaryInfC{$\langle \PL{\textbf{if} } e \PL{ \textbf{then} } s_1 \PL{ \textbf{else} } s_2 \PL{ \textbf{end}}, \sigma\rangle  \step \langle s_1, \tau\rangle $}
       \DisplayProof
       \hfill
       \AxiomC{$\langle e, \sigma\rangle  \mstep \langle\emptyP, \tau\rangle $}
       \AxiomC{$\mathit{bool}(\tau(\result), \mathrm{false})$}
       \BinaryInfC{$\langle \PL{\textbf{if} } e \PL{ \textbf{then} } s_1 \PL{ \textbf{else} } s_2 \PL{ \textbf{end}}, \sigma\rangle  \step \langle s_2, \tau\rangle $}
       \DisplayProof

 \item \AxiomC{$\langle e, \sigma\rangle  \mstep \langle\emptyP, \tau\rangle $}
       \AxiomC{$\tau(\result) = \Vnull$}
       \BinaryInfC{$\langle \PL{\textbf{if} } e \PL{ \textbf{then} } s_1 \PL{ \textbf{else} } s_2 \PL{ \textbf{end}}, \sigma\rangle  \step \langle \emptyP, \mathbf{fail}\rangle $}
       \DisplayProof

 \item \AxiomC{$\langle e, \sigma\rangle  \mstep \langle\emptyP, \tau\rangle $}
       \AxiomC{$\neg \mathit{bool}(\tau(\result))$}
       \BinaryInfC{$\langle \PL{\textbf{if} } e \PL{ \textbf{then} } s_1 \PL{ \textbf{else} } s_2 \PL{ \textbf{end}}, \sigma\rangle  \step \langle \emptyP, \mathbf{typeerror}\rangle $}
       \DisplayProof

 \item $\langle \PL{\textbf{while} } e \PL{ \textbf{do} } s \PL{ \textbf{done}}, \sigma\rangle  \step \langle \PL{\textbf{if} } e \PL{ \textbf{then} } s; \PL{\textbf{while} } e \PL{ \textbf{do} } s \PL{ \textbf{done} \textbf{else} null \textbf{end}}, \sigma\rangle$

 \item $\langle\Vector{u} := \Vector{v}, \sigma\rangle  \step \langle\emptyP, \sigma[\Vector{u} := \sigma(\Vector{v})]\rangle $ \info{where $\Vector{u}, \Vector{v} \in \Var^+_L$}

 \item $\langle\PL{\textbf{begin local} } \Vector{u} := \Vector{v}; S \PL{ \textbf{end}}, \sigma\rangle  \step \langle\Vector{u} \Vector{\overline{u}} := \Vector{v} \Vector{\uninitialized}; S; \Vector{u} \Vector{\overline{u}} := \sigma(\Vector{u}\Vector{\overline{u}}), \sigma\rangle$

 \info{where $\{\Vector{\overline{u}}\} = \Var_L \setminus (\{\Vector{u}\} \cup \Var_S)$ and $\Vector{\uninitialized}$ is a fitting sequence of $\uninitialized$ values.}

 \item \AxiomC{$\langle e_i, \sigma_i\rangle \mstep \langle \emptyP, \sigma_{i+1}\rangle $ for all $i \in \mathbb{N}_n$}
       \UnaryInfC{$\langle e_0.m(e_1,...,e_n), \sigma_0\rangle  \step \langle\PL{\textbf{begin local} } \self, \Vector{u} := \sigma_1(\result),...,\sigma_{n+1}(\result); s \PL{ \textbf{end}}, \sigma_{n+1}\rangle$}
       \DisplayProof
       
 \info{where $\sigma_1(\result) \not= \Vnull$, $\mathbf{method}\PL{ } m(u_1,...,u_n) \{ s \} \in \Method_C$ and $\sigma_1(\result.@\mathbf{c}) = \oC_C$.}
 
 \item \AxiomC{$\langle e_0, \sigma_0\rangle  \mstep \langle\emptyP, \sigma_1\rangle $}
       \AxiomC{$\sigma_1(\result) = \Vnull$}
       \BinaryInfC{$\langle e_0.m(e_1,...,e_n), \sigma_0\rangle  \step \langle\emptyP, \mathbf{fail}\rangle $}
       \DisplayProof
       \hfill
       \AxiomC{$\langle e_0, \sigma_0\rangle  \mstep \langle\emptyP, \sigma_1\rangle $}
       \AxiomC{$\sigma_1(\result) \not= \Vnull$}
       \BinaryInfC{$\langle e_0.m(e_1,...,e_n), \sigma_0\rangle  \step \langle\emptyP, \mathbf{typeerror}\rangle $}
       \DisplayProof

 \hfill \info{where $\sigma_1(\result.@\mathbf{c}) = \oC_C$ and $\not\exists \mathbf{method}\PL{ } m(u_1,...,u_n) \{ s \} \in \Method_C$}

 \item $\langle\mathbf{new}\PL{ } C(e_1,...,e_n), \sigma\rangle  \step \langle\mathbf{new}_C.init(e_1,...,e_n), \sigma\rangle $

 \item $\langle\mathbf{new}_C, \sigma\rangle  \step \langle\emptyP, \sigma[o := \mathit{init}_C][\result := o]\rangle $ \info{where $o = \creation(\sigma)$}

 \item \AxiomC{$\begin{matrix}
                 \langle e_i, \sigma_i\rangle \mstep \langle\emptyP, \sigma_{i+1}\rangle$ for all $i \in \{0,1\} \\
                 \exists b:\typeALbool \bullet b \leftrightarrow \sigma_1(\result) = \sigma_2(\result) \wedge s_r \equiv true \wedge b \vee s_r \equiv false \wedge \neg b
                \end{matrix}$}
       \UnaryInfC{$\langle e_0 == e_1, \sigma_0 \rangle \step \langle s_r, \sigma_2 \rangle$}
       \DisplayProof
 \item \AxiomC{$\langle e, \sigma\rangle \mstep \langle\emptyP, \sigma_1\rangle$,
               $\exists b:\typeALbool \bullet b \leftrightarrow \sigma_1(\result.@\mathbf{c}) = \oC_C$,
               $s_r \equiv true \wedge b \vee s_r \equiv false \wedge \neg b$}
       \UnaryInfC{$\langle e \; is\_a? \; C, \sigma \rangle \step \langle s_r, \sigma_1 \rangle$}
       \DisplayProof
\end{enumerate}
%\caption{\label{fig:dyn_op_sem} \ifx \version\conference Excerpt from \langDyn's structural operational semantics. See \cite{extended_version} for the full version. \else \langDyn's structural operational semantics.\fi}
\caption{\label{fig:dyn_op_sem} \langDyn's structural operational semantics.}
\end{figure}
\todo{add constants $\oC{C}$ for all classes $C$ to the assertion language}
\todo{add rules for try..catch-blocks}
Note how the rule for assignment uses the two-step idea to handle side-effecting expressions.
The rules for conditionals and while loops also use it to evaluate
the condition first and then branch on its result. Since no type system guarantees this result to be
boolean, further distinguished behaviors for failures and type errors are
necessary. The same holds for receivers of method calls.

Additionally, the rules for method call (or better: \textbf{begin local}-blocks)
and object creation instantiate all local- and instance variables to $\uninitialized$,
which marks them as ``not yet created'' and causes $\mathbf{typeerror}$ in the rule of variable access.

Note also the handling of special variables in method calls: on entry, $\self$ is set to the receiver of the
method call while on exit $\result$ intentionally remains unmodified to pass the return value back to the
caller.

\ifx \version\extended
Constructors are normal methods conventionally named $\mathit{init}$ that are called on newly created
instances
directly after they were created. The instance creation (activation) itself is called $\textbf{new}_C$.
Note that \textbf{new} $C(...)$ returns the constructor's return value
which is not necessarily the newly created instance. Also note that calling \textbf{new} $C'(...)$ for a class $C'$
that does not have a method $init$ results in a $\mathbf{typeerror}$.
\fi

\section{Axiomatic Semantics}
\label{sec:dyn:axiomatic}

\subsection{Tagged Hoare Logic}
% - clarity: Triples always carry their notion of correctness with them
% - brevity: no need to state the rules separately for each notion of correctness
% - tool support: notion of correctness should be the decision of the user, not the tool developer
% - rules connecting different notions of correctness (Decomposition rule, try..catch)

The original paper of Hoare \cite{Hoa69} considers partial correctness. 
Other ``notions of correctness'' like strong partial correctness and total correctness were added
later as separate proof systems.
While termination as a liveness property might justify this special handling,
there seem to be little reason to grant this special place also to properties like failsafety and
typesafety. They do, however, affect the proof rules (mostly by adding additional preconditions) and
hence triggered the creation of new proof systems for new ``notions of correctness''. Additionally, the
term ``total correctness'' was interpreted as ``the absence of any kind of fault'' and hence 
strongly depends on what other faults the authors are considering.
Furthermore, in this abundance of available proof systems, tool designers are forced to
choose which one to implement, depriving their users of the choice which properties
they actually want to verify. From a tool-design perspective, it would be much better to make all
properties part of the specification, have a single proof system dealing with them and
allowing the users to choose which guarantees to derive for which part of the program.
We hence propose the formalism of \emph{tagged Hoare logic}, a uniform framework for all these
properties featuring a single proof system to treat them.

A (big step) program semantics maps programs and initial states to sets of final states.
% This mapping is canonically extended to sets of states.
Traditionally, each notion of correctness needs its own program semantics as they differ in what
characteristics of a computation they guarantee. We define the (infinite) set of
(finite or infinite) computations as 
$Comp = \mathit{Conf}^* \cup \mathit{Conf}^\omega$ and those of a program $s$
starting in an initial state $\sigma$ as
\[Comp(s,\sigma) = \{ C_0,C_1,... \mid C_0 = \langle s,\sigma \rangle \wedge \forall i \bullet C_i \rightarrow C_{i+1} \} \subset Comp\;.\]
\noindent We use the symbol $\rho$ to denote elements of $Comp$ and define the following tags along with their
respective error states:

$\tags = \{ \mtagterminates, \mtagtypesafe, \mtagfailsafe \}$, $\Sigma_\bottom = \{\bottom, \mathbf{typeerror}, \mathbf{fail}\}$,

$\pitchfork: \tags \mapsto \Sigma_\bottom$,  $\Sigma_+ = \Sigma \uplus \Sigma_\bottom$

$\pitchfork(\mtagterminates) = \bottom$, $\pitchfork(\mtagtypesafe) = \mathbf{typeerror}$, $\pitchfork(\mtagfailsafe) = \mathbf{fail}$

\noindent Their behaviour may be defined as a selector for their respective characteristic:

$\selector: \tags \cup \{\emptyset\} \mapsto Comp \mapsto \mathcal{P}(\Sigma_+)$

$\selector_{\emptyset}(\rho) = \begin{cases}
                    \{ \tau \} & \text{ if } \rho = C_0,...,C_n \wedge C_n = \langle \emptyP,  \tau \rangle \wedge \tau \in \Sigma \\
		    \{ \} & \text{ otherwise}
                   \end{cases}
$

$\selector_{\mtagterminatessmall}(\rho) = \begin{cases}
                                             \{ \bottom \} & \text{ if } \rho \text{ is infinite} \\
					     \{\} & \text{ otherwise}
                                            \end{cases}
$

$\selector_{\mtagtagsmall}(\rho) = \begin{cases}
                                             \{ \pitchfork(\mtagtag) \} & \text{ if } \rho = C_0,...,C_n \wedge C_n = \langle \emptyP, \pitchfork(\mtagtag) \rangle \\
					     \{\} & \text{ otherwise}
                                   \end{cases}
$\newline
for all other tags.
%
% $\selector_{\mtagfailsafe}(\rho) = \begin{cases}
%                                              \{ \mathbf{fail} \} & \text{ if } \rho = C_0,...,C_n \wedge C_n = \langle \emptyP,  \mathbf{fail} \rangle \\
% 					     \{\} & \text{ otherwise}
%                                    \end{cases}
% $
%
Finally, we are able to define tagged program semantics
\[
\semantics: \mathcal{P}(\tags) \mapsto Prog \mapsto \Sigma \mapsto \mathcal{P}(\Sigma_+)
\]
allowing arbitrary combinations of correctness notions. Let $\mtagvarX \subseteq \tags$, then
\[\semantics_{\mtagvarXsmall}\llbracket s \rrbracket(\sigma) = \bigcup \{ \selector_{\mtagtagsmall}(\rho) \mid \rho \in Comp(s,\sigma), \mtagtag \in \mtagvarX \cup \{ \emptyset \} \}\]
\noindent which is certainly the most central ingredient of a Tagged Hoare Logic. However, we first need to extend
the semantics of our assertions to also include tags
\[ \llbracket p \rrbracket_{\mtagvarXsmall} = \{ \sigma \mid \sigma \in \Sigma  \wedge \sigma, \mtagvarX \models p \} \]
before we can properly define the meaning of our Tagged Triples:

\begin{definition}[Tagged Hoare Triples]
%A \emph{Tagged Hoare Triple}

%$\{p\} s \{\mtagvarX \wedge q\}$ is true (written

$\models \{p\} s \{\mtagvarX \wedge q\}$ \hfill iff \hfill
$\semantics_{\mtagvarXsmall}\llbracket s \rrbracket(\llbracket p \rrbracket_{\mtagvarXsmall}) \subseteq \llbracket q \rrbracket_{\mtagvarXsmall}$

\noindent where $\models$ denotes semantic truth of \emph{Tagged Hoare Triples}.

\end{definition}

\subsection{Assertion Language}

Before going into details of the program logic, we introduce the assertion language \langAL{}. Its syntax
is depicted in Figure~\ref{fig:dyn_ass_lang}. Essentially, it is predicate logic
with quantification over finite sequences of typed elements -- weak second order logic.
We extend the logic with constants $c_{\pure}$ and operations $op(\Vector{l})$
corresponding to \textbf{dyn}'s syntactic sugar for boolean values, natural numbers, strings and lists
(which includes the usual arithmetic operations on both booleans and natural numbers).
Also, $c_{\pure}$ contains constants $\oC_C$ denoting the representative objects of all classes $C \in \Class$.
Note that our assertion language is statically typed, as usual. Its type system however is simplistic:
basic types $\BaseTypeX = \{\typeALnum, \typeALobj, \typeALbool, \typeALstring\}$ form a flat
lattice with $\top$ and $\bottom$ and a type constructor $\tau^{*}$ for finite sequences of
elements of type $\tau$.

Assertions contain typed logical expressions ($l$).
Such expressions consist of accesses to logical variables (of some type $t \in \TypeAL$),
local program variables (of type $\mathbb{O}$)
including the self-reference $\self$, instance variables ($l.@x$ where both $l$ and the result are 
of type $\mathbb{O}$), typed constants and typed operations.
Note that contrary to programming expressions, logical expressions are able to access
instance variables of objects other than $\self$.

Assertions are then constructed from equations between logical expressions of identical type, boolean connectives and
quantification over finite sequences.% of elements of basic types.

Following \cite{DeBoerPierik2003},
undefined operations like dereferencing a $\Vnull$ value or accessing a sequence with an index
that is out of bounds ($l[n]$ with $n \ge |l|$) yield a $\Vnull$ value and equality is non-strict with respect
to such values ($\Vnull = \Vnull$ is $\mathit{true}$) in order to keep assertions two-valued.
Also, for logical expressions $l \in LExp$, we extend the state-access to $\sigma(l)$ in the
canonical way.\todo{also give a semantics for the assertions! at least: $L.@v$}

%\todo{this paragraph neccessary? don't we refer to the appendix for substitutions anyway?}
%The substitution $[\PL{u} := \PL{new}]$ is equivalent to the substitution $[\PL{u}/\PL{new}]$
%defined in \cite{DeBoerPierik2002CompAided}. It computes the weakest precondition of an object creation
%statement $\PL{u} := \PL{new}$, thereby handling references to the newly created
%object as well as quantification over objects.

%Its signature can, however, be extended.
%%TODO{mention conditions for purity?}

To link programming language-objects with assertion-values,
we define
\begin{definition}[Mapping Predicates]\footnote{The predicate $\mapnum(o,n)$ is recursive. However, the
technique used for proving the case for primitive recursion in 
\ifx \version\extended
Lemma~\ref{lem:trans_mu_rec_to_ass}
\else
\cite[Lemma 5]{extended_version}
\fi
allows expressing it in \langAL.}
\footnote{$@\mathrm{pred}$ and $@\mathrm{to\_ref}$ are instance variables of the classes num and bool respectively.} ($\forall o: \typeALobj$, $n: \typeALnum$, $b: \typeALbool$)

$ \mapnum(o) \equiv o \not= \Vnull \wedge o.@\mathbf{c} = \oC_{num} $

$ \mapnum(o,n) \equiv \mapnum(o) \wedge n = 0 \rightarrow o.@\mathrm{pred} = \Vnull \wedge n > 0 \rightarrow \mapnum(o.@\mathrm{pred},n-1)$
\vspace{0.1cm}

$\mapbool(o) \equiv o \not= \Vnull \wedge o.@\mathbf{c} = \oC_{bool} \hspace{1cm}
 \mapbool(o,b) \equiv \mapbool(o) \wedge b \leftrightarrow o.@\mathrm{to\_ref} \not= \Vnull$
\end{definition}

%for all $o$ of type $\typeALobj$, $b$ of type $\typeALbool$.
To see that mapping predicates are necessary for completeness, consider the intermediate assertion $p$ in the following program

\vspace{0.1cm}
\noindent$P \equiv \mathbf{if}\PL{ }b\PL{ \textbf{then} }x := 5\PL{ \textbf{else} }x := \mathit{true}\PL{ }\mathbf{end} \{p\};$

\noindent\hspace{0.75cm}$ \mathbf{if}\PL{ } x \PL{ is\_a? } bool \PL{ \textbf{then} }\mathbf{if}\PL{ }x\PL{ \textbf{then} }x := 10\PL{ }\mathbf{end}\PL{ \textbf{else} }x := x * 2\PL{ \textbf{end}}$
\vspace{0.1cm}

Since \langAL{} is statically typed, we must also give a type to the program variable $x$.
Now, giving it the type $\typeALnum$ would allow us to express $x = 5$, but not $x = true$ while giving it the
type $\typeALbool$ raises the converse problem. However, using mapping predicates, it is possible to
accurately describe the set of intermediate states as $\mapnum(x,5) \vee \mapbool(x,true)$. From this observation it
is not hard to see that $\{true\} P \{ x = 10 \}$ (or $\{true\} P \{ \mapnum(x,10) \}$) is not derivable
without mapping predicates.
We use the notation $\sigma, \mtagvarX \models a$ to denote the fact that the assertion $a$ is true in
the state $\sigma$ under the tags $\mtagvarX$. The definition of $\models$ is
standard except for the case
\[ \sigma, \mtagvarX \models \mtagtag  \text{ iff } \mtagtag \in \mtagvarX. \]
\begin{figure}
$\mathit{Asrt} \ni a ::= l = l \mid \llbracket l \rrbracket \in \{Cl\} \mid \neg a \mid a \wedge a \mid \exists v : t^* \bullet a \mid \mtagtag$  \hfill $t \in \TypeAL$, $\mtagtag \in \tags$

$\mathit{LExp} \ni l ::= v \mid \PL{u} \mid l.\PL{@x} \mid \Vnull \mid \uninitialized \mid \self \mid$ \textbf{if} $l$ \textbf{then} $l$ \textbf{else} $l$ \textbf{end} $\mid c_{\pure} \mid op(\Vector{l})$

$\{+,-,*,div,mod,<,=,\wedge,\neg,|\cdotp|,\cdotp[n]\} \subseteq op$ \hfill (brackets are used for disambiguation)

$Cl ::= \epsilon \mid CL$ \hfill $CL ::= C \mid C, CL$ \hfill $C \in \Class$

with the usual abbreviations:
$a_1 \vee a_2 \equiv \neg(\neg a_1 \wedge \neg a_2)$, $a_1 \rightarrow a_2 \equiv \neg a_1 \vee a_2$,
$a_1 \leftrightarrow a_2 \equiv a_1 \rightarrow a_2 \wedge a_2 \rightarrow a_1$, $\forall v:t^* \bullet a \equiv \neg \exists v:t^* \bullet \neg a$,
$true \equiv (\Vnull = \Vnull)$, $false \equiv \neg true$,
$Q v:t \bullet a \equiv Q v:t^* \bullet |v| = 1 \wedge a[v[0]/v]$ for $Q \in \{\forall, \exists\}$,
$l \equiv l = true$ if $l$ is of type $\typeALbool$.

\caption{\label{fig:dyn_ass_lang} Syntax of the assertion language \langAL{}}
\end{figure}

\todo{define $op$ as natural numbers + arithmetic}
\todo{give types for the operations in the extended version}
\todo{consistently use verbose x for program variables and $v$ for logical variables}

\subsection{(Tagged) Hoare Logic for Dynamically Typed Programs}

Our exposition of the proof rules of $\PSHL$ will use three substitutions on assertions.
Proper definitions for all three can be found in 
\ifx \version\extended
Appendix~\ref{app:substitutions}.
\else
\cite[Appendix B]{extended_version}.
\fi

The special variable $\result$ may appear in both pre- and postconditions. 
In preconditions it references some initial value,
in postconditions the return value of the last executed expression.
Note that it is important that $\result$ can appear in preconditions. Otherwise the weakest
precondition $WP(\result, \result = null)$ would not be expressible which would induce incompleteness.

For a \langDyn{} statement $s$ let $var(s)$ ($change(s)$) denote the set of
variables accessed in $s$ (appearing on the left of an assignment in $s$).
For an assertion $p$ let $free(p)$ denote the set of free variables of $p$
and $p[v := t]$ the result of substituting $t$ for $v$ in $p$.

% Text as is was before:
%
% The special variable $\result$ always references the return value of the last executed expression. It may
% appear both in pre- and in postconditions. While it seems unreasonable to use $\result$
% in preconditions, as it is not clear which statement's return value it refers to, restricting the use of $\result$ to
% postconditions introduces incompleteness as the weakest precondition of the (empty) program $\result$ with
% respect to the postcondition $\result = null$ would not be expressible any more.

\vspace{0.1cm}
% \noindent AXIOM: CONST
% 
% \begin{center}
% $\{p[\result := \Vnull]\} \PL{null} \{\mtagvarX \wedge p\}$
% \end{center}
%
\noindent AXIOM: VAR \hfill VAR-TAG

\begin{center}
$\{p[\result := v]\} v \{p\}$ \hfill $\{\mtagtypesafe \rightarrow v \not= \uninitialized\} v \{\mtagvarX\}$
\end{center}

\info{Note: includes the case of $v \equiv \self$.}

\noindent AXIOM: IVAR \hfill IVAR-TAG

\begin{center}
$\{p[\result := \self\PL{.@v}]\} @v \{p\}$ \hfill $\{ \mtagtypesafe \rightarrow \self\PL{.@v} \not= \uninitialized \} @v \{ \mtagvarX \}$
\end{center}

\noindent RULE: ASGN (both normal and instance variables) \hfill AXIOM: CONST
\begin{center}
\AxiomC{$\{p\} e \{ \mtagvarX \wedge q[v := \result] \}$}
\RightLabel{where $v \in \Var$}
\UnaryInfC{$\{p\} v := e \{ \mtagvarX \wedge q\}$}
\DisplayProof
\hfill
$\{p[\result := \Vnull]\} \PL{null} \{\mtagvarX \wedge p\}$
\end{center}
\noindent RULE: SEQ

\begin{center}
\AxiomC{$\{p\} s_1 \{ \mtagvarX \wedge r\}$}
\AxiomC{$\{r\} s_2 \{ \mtagvarX \wedge q\}$}
\BinaryInfC{$\{p\} s_1; s_2 \{ \mtagvarX \wedge q\}$}
\DisplayProof
\end{center}

\noindent RULE: COND
\begin{center}
\AxiomC{$\begin{matrix}
          \{p\} e \{\mtagvarX \wedge r \wedge \mtagfailsafe \rightarrow \result \not= \Vnull \wedge \mtagtypesafe \rightarrow \result \not= \Vnull \rightarrow \mapbool(\result) \} \\
          \{r \wedge \mapbool(\result,true)\} s_1 \{\mtagvarX \wedge q\} \\
          \{r \wedge \mapbool(\result,false)\} s_2 \{\mtagvarX \wedge q\}
         \end{matrix}$}
       
\UnaryInfC{$\{p\}$ \textbf{if} $e$ \textbf{then} $s_1$ \textbf{else} $s_2$ \textbf{end} $\{\mtagvarX \wedge q\}$}
\DisplayProof
\end{center}

\noindent RULE: LOOP
\begin{center}
\AxiomC{$\begin{matrix}
          \{p\} e \{ \mtagvarX \wedge p' \wedge \mtagfailsafe \rightarrow \result \not= \Vnull \wedge \mtagtypesafe \rightarrow \result \not= \Vnull \rightarrow \mapbool(\result) \} \\
          \{p' \wedge \mapbool(\result, true) \} s \{ \mtagvarX \wedge p \} \\
          \{p' \wedge \mapbool(\result, true) \wedge r(z) \} s;e \{p' \wedge \mtagterminates \rightarrow \forall z':\typeALnum \bullet r(z') \rightarrow z' < z \}
%          \mtagterminates \in \mtagvarX \rightarrow p \vee p' \rightarrow \forall z':\typeALnum \bullet r(z') \rightarrow z' \ge 0
         \end{matrix}$}
         %TODO: find a better way to express $\mtagterminates \in \mtagvarX$
\UnaryInfC{$\{p\}$ \textbf{while} $e$ \textbf{do} $s$ \textbf{done} $\{ \mtagvarX \wedge p'[\result := b] \wedge \mapbool(b, false) \wedge \result = \Vnull\}$}
\DisplayProof
\end{center}
\info{where $b$ is a logical variable of type $\typeALbool$,
$z$ is a logical variable of type $\typeALnum$ that does not appear in $p,p',e$ or $s$,
$r(z)$ is a predicate with $z$ among its free variables such that $\forall \sigma \bullet \sigma \models p' \rightarrow \exists z':\typeALnum \bullet r(z')$ and $r(z')$ is the result of substituting $z'$ for $z$ in $r(z)$.}

\noindent RULE: CONS
\begin{center}
\AxiomC{$p \rightarrow p_1, \{p_1\} s \{ \mtagvarY \wedge q_1\}, q_1 \rightarrow q, \mtagvarY \supseteq \mtagvarX$}
\UnaryInfC{$\{p\} s \{ \mtagvarX \wedge q\}$}
\DisplayProof
\end{center}

% \noindent RULE: PASGN
% %
% \begin{center}
% $\{p[\Vector{u} := \Vector{t}]\} \Vector{u} := \Vector{t} \{\mtagvarX \wedge p\}$
% \end{center}
% %
% where $\{\Vector{u}\} \subseteq \Var_L$ and $\{\Vector{t}\} \subseteq \Var_L \cup \{\PL{null}\}$.

\noindent RULE: BLCK \hfill AXIOM: PASGN
\begin{center}
\AxiomC{$\begin{matrix}
          \{p\} \Vector{u}\Vector{\overline{u}} := \Vector{t}\Vector{\uninitialized}; s \{\mtagvarX \wedge q\} \\
          (\Var_L \setminus \{\result\}) \cap \free(q) = \emptyset
         \end{matrix}$}
\UnaryInfC{$\{p\}$ \textbf{begin local} $\Vector{u} := \Vector{t}; s \text{ \textbf{end}} \{\mtagvarX \wedge q\}$}
\DisplayProof
\hfill
$\{p[\Vector{u} := \Vector{t}]\} \Vector{u} := \Vector{t} \{\mtagvarX \wedge p\}$
\end{center}
\info{where $\{\Vector{u}\} \subseteq \Var_L$ and $\{\Vector{t}\} \subseteq \Var_L \cup \{\PL{null}\}$,
$\{\Vector{\overline{u}}\} = \Var_L \setminus (\{\Vector{u}\} \cup \Var_S)$ and $\Vector{\uninitialized}$ is a fitting sequence of $\uninitialized$ constants.}

\noindent RULE: METH
%
% \begin{center}
% \AxiomC{$\begin{matrix}
%           \{p_i\} E_i \{p_{i+1} \wedge \result = l_i \}\text{ for }i \in \mathbb{N}_n \\
%           \{p_{n+1}\} l_0.m(l_1,...,l_n) \{q\}
%          \end{matrix}$}
%        
% \UnaryInfC{$\{p_0\} E_0.m(E_1,...,E_n) \{q\}$}
% \DisplayProof
% \end{center}
\begin{center}
\AxiomC{$\begin{matrix}
          \{p_i\} e_i \{ \mtagvarX \wedge p_{i+1}[v_i := \result] \}\text{ for }i \in \mathbb{N}_n \\
          \{p_{n+1}\} v_0.m(v_1,...,v_n) \{ \mtagvarX \wedge q\}
         \end{matrix}$}
       
\UnaryInfC{$\{p_0\} e_0.m(e_1,...,e_n) \{ \mtagvarX \wedge q\}$}
\DisplayProof
\end{center}
\info{where the $v_i$ are fresh local variables that do not occur in any $e_j$ for all $i,j \in \mathbb{N}_n$.}

\noindent RULE: REC
\begin{center}
\AxiomC{$\begin{matrix}
          A \vdash \{p\} s \{\mtagvarX \wedge q\}, \\
          A' \vdash \{p_i \wedge r_i(z)\} \PL{\textbf{begin local} \self}, \Vector{u_i} := l_i,\Vector{v_i}; s_i \text{ \textbf{end}} \{\mtagvarX_i \wedge q_i\}, i \in \mathbb{N}^1_n \\
          p_i \rightarrow (\mtagfailsafe \rightarrow l_i \not= \Vnull \wedge \mtagtypesafe \rightarrow l_i \not= \Vnull \rightarrow l_i.@\mathbf{c} = \oC_{C_i}), i \in \mathbb{N}^1_n
         \end{matrix}$}
       
\UnaryInfC{$\{p\} s \{\mtagvarX \wedge q\}$}
\DisplayProof
\end{center}
\info{where \textbf{method} $m_i(\Vector{u_i}) \{ s_i \} \in \Method_{C_i}$,
$A = \{ \{p_i\} l_i.m_i(\Vector{v_i}) \{\mtagvarX_i \wedge q_i\} \mid i \in \mathbb{N}^1_n \}$,
$A' = \{ \{p_i \wedge (\mtagterminates \rightarrow \forall z':\typeALnum \bullet r_i(z') \rightarrow z' < z)\} l_i.m_i(\Vector{v_i}) \{\mtagvarX_i \wedge q_i\} \mid i \in \mathbb{N}^1_n \}$,
$z$ is a logical variable of type $\typeALnum$ that does not occur in $p_i,q_i$ and $s_i$ for $i \in \mathbb{N}^1_n$ and
is treated in the proofs as a constant,
$r_i(z)$ for $i \in \mathbb{N}^1_n$ are predicates with $z$ among their free variables such that $\forall \sigma \bullet \sigma \models p_i \rightarrow \exists z':\typeALnum \bullet r_i(z')$
for all $i \in \mathbb{N}^1_n$
and $r_i(z')$ denotes the result of substituting $z'$ for $z$ in $r_i(z)$.}
% \noindent RULE: REC (update)\todo{remove when done}
% %
% (partial correctness \Optional{ \StrongPC{+strong} \TypesafePC{+typesafe}} )
% %
% \begin{center}
% \AxiomC{$I_1,...,I_k$}
% \AxiomC{$A_1,...,A_k \vdash F_0,...,F_k$}
% \BinaryInfC{$\{p\} S \{q\}$}
% \DisplayProof
% \end{center}
% where
% \begin{itemize}
%  \item $A_i \equiv \{p_i[\Vector{u_i} := \Vector{v_i}][\Vector{f_i} := \Vector{z_i}]\} l_i.m_i(\Vector{v_i}) \{r_i\}$
%  \item $F_0 \equiv \{p\} S \{q\}$
%  \item $F_i \equiv \{p_i\} S_i \{q_i[\result := \mathit{return}_i]\}$ for $i \in \mathbb{N}^1_k$
%  \item $I_i \equiv q_i[\Vector{f_i} := \Vector{z_i}] \rightarrow r_i[\mathit{return}_i := \result]$
% \end{itemize}
% 

\noindent AXIOM: EQUAL \hfill AXIOM: IS\_A

$\{ true \} \PL{u}_0 == \PL{u}_1 \{ \mtagvarX \wedge \mapbool(\result,\PL{u}_0 = \PL{u}_1) \}$

\hfill $\{ true \} \PL{u}_0 \; is\_a? \; C \{ \mtagvarX \wedge \mapbool(\result,\llbracket \PL{u}_0 \rrbracket \in \{C\}) \}$

\noindent RULE: CNST \hfill AXIOM: NEW
\begin{center}
\AxiomC{$\{p\} \mathbf{new}_C.init(\Vector{e}) \{\mtagvarX \wedge q\}$}
\UnaryInfC{$\{p\} \mathbf{new}\text{ } C(\Vector{e}) \{\mtagvarX \wedge q\}$}
\DisplayProof
\hfill $\{p[\result := \mathbf{new}_C]\} \mathbf{new}_C \{\mtagvarX \wedge p\}$
\end{center}

% \noindent AXIOM: NEW
% \begin{center}
% $\{p[\result := \mathbf{new}_C]\} \mathbf{new}_C \{\mtagvarX \wedge p\}$
% \end{center}

\ifx \version\extended
\subsection{Auxiliary Rules}
\label{sec:aux_rules}

\noindent RULE: DISJ

\begin{center}
\AxiomC{$\{p\} s \{\mtagvarX \wedge q\}$}
\AxiomC{$\{r\} s \{\mtagvarX \wedge q\}$}
\BinaryInfC{$\{p \vee r\} s \{\mtagvarX \wedge q\}$}
\DisplayProof
\end{center}

\noindent RULE: CONJ

\begin{center}
\AxiomC{$\{p_1\} s \{\mtagvarX \wedge q_1\}$}
\AxiomC{$\{p_2\} s \{\mtagvarY \wedge q_2\}$}
\BinaryInfC{$\{p_1 \wedge p_2\} s \{\mtagvarX \wedge \mtagvarY \wedge q_1 \wedge q_2\}$}
\DisplayProof
\end{center}

\noindent RULE: $\exists$-INT

\begin{center}
\AxiomC{$\{p\} s \{\mtagvarX \wedge q\}$}
\UnaryInfC{$\{\exists x. p\} s \{\mtagvarX \wedge q\}$}
\DisplayProof
\end{center}
where $x \not\in var(\Method) \cup var(s) \cup \free(q)$.

\noindent RULE: INV

\begin{center}
\AxiomC{$\{r\} s \{\mtagvarX \wedge q\}$}
\UnaryInfC{$\{p \wedge r\} s \{\mtagvarX \wedge p \wedge q\}$}
\DisplayProof
\end{center}
where $\free(p) \cap (change(\Method) \cup change(s)) = \emptyset$ and $p$ does not contain quantification
over objects.
\todo{explain $change(\Method) \subseteq \Var_I$}

\noindent RULE: SUBST

\begin{center}
\AxiomC{$\{p\} s \{\mtagvarX \wedge q\}$}
\UnaryInfC{$\{p[\Vector{z} := \Vector{t}]\} s \{\mtagvarX \wedge q[\Vector{z} := \Vector{t}]\}$}
\DisplayProof
\end{center}
where $var(\Vector{z}) \cap (var(\Method) \cup var(s)) = var(\Vector{t}) \cap (change(\Method) \cup change(s)) = \emptyset$.
\else
The auxiliary rules as well as some others are mostly standard and hence omitted. Their tagged versions are given in \cite{extended_version}, however.
\fi
% weakening rule is only for partial correctness, not for strong partial correctness

The fact that dyn-expressions have side effects is mirrored in several rules:
Like their corresponding rules in the operational semantics, the usual axiom for assignment is
turned into a rule and the COND and LOOP rules both evaluate the condition before branching on its result in
an intermediate state.

The rules
PASGN, BLCK, METH and REC are needed to handle method calls. After handling side effecting expressions
in arguments beforehand (METH) and ensuring that methods are only called on receivers supporting them (last premise of REC),
method calls are assumed to satisfy the same properties as a block executing the body of the called method
in an environment with local variables suitably initialized by parallel assignment (BLCK,PASGN).

The rules CNST and NEW handle object creation using the respective substitution defined in 
\ifx \version\extended
appendix~\ref{app:substitutions}.
\else
\cite[appendix B]{extended_version}.
\fi

The LOOP and REC rules feature a novel form of loop variants / recursion bound. The basic idea is to
use a predicate $r(z)$ instead of the usual integer expression $t$ in order to allow quantification within
loop variants / recursion bounds. While this was primarily introduced to circumvent a common
incompleteness issue in Hoare logics for total correctness (see proof of Theorem~\ref{thm:comp_rec_meth}
for details), note that it also allows using mapping predicates directly in loop variants / recursion bounds,
i.e. proving
\[\{\mapnum(i)\} \PL{\textbf{while} } i > 0 \PL{ \textbf{do} } i := i - 1 \PL{ \textbf{done}} \{\mtagterminates\}\]
with $r(z) \equiv \mapnum(i,z)$.

% no translational approach this time
%\section{The Translational Approach}
%\label{sec:translational_approach}

%\section{Target Language}
%\label{sec:target_lang}

%\subsection{Operational Semantics}

%\subsection{Axiomatic Semantics}

% - target language
%   - operational semantics for target language
%   - axiomatic semantics for target language
%   - reference soundness & completeness results
%   ? add completeness proof for OO programs with weak 2nd-order quantification over objects to appendix

\section{Soundness}
\label{sec:soundness}
\ifx \version\extended
\todo{exemplarily proof soundness of some rules against op. Semantics}

Soundness follows from a standard inductive argument. We will only present the case for
the LOOP rule as the idea of using a predicate $r$ as a loop variant for total correctness is novel.

\noindent \textbf{Induction Hypothesis:}
$\vdash \{p\} s \{\mtagvarX \wedge q\} \rightarrow \models \{p\} s \{\mtagvarX \wedge q\}$ for all assertions $p$ and $q$ and all \langDyn{}
statements $s$.

%\noindent \textbf{Induction Basis:}
\noindent \textbf{Induction Step:}

\noindent \textbf{Partial Correctness:}
Given $\vdash \{p\} e \{\mtagvarX \wedge p'\}$ and $\vdash \{p' \wedge \mapbool(\result, true)\} s \{\mtagvarX \wedge p\}$,
by the induction hypothesis $\models \{p\} e \{\mtagvarX \wedge p'\}$ and
$\models \{p' \wedge \mapbool(\result, true)\} s \{\mtagvarX \wedge p\}$ follow.

Hence, when executing the program $\PL{\textbf{while} } e \PL{ \textbf{do} } s \PL{ \textbf{done}}$ in a state $\sigma \models p$, the
operational semantics will first apply rule 9 yielding the configuration
$\langle \PL{\textbf{if} } e \PL{ \textbf{then} } s; \PL{\textbf{while} } e \PL{ \textbf{do} } s \PL{ \textbf{done}} \PL{ \textbf{else} null \textbf{end}}, \sigma \rangle$,
then apply whatever rules neccessary to evaluate $\langle e, \sigma \rangle$ to a final configuration $\langle \emptyP, \tau \rangle$.
From $\models \{p\} e \{p'\}$ we can deduce $\tau \models p'$.
Furthermore, the operational semantics uses rules 6-8 to branch on the value of $\tau(\result)$.
Now, for the case of partial correctness, we are only interested in normal program termination,
the cases yielding $\mathbf{fail}$ or $\mathbf{typeerror}$ will be handled below.
Hence there are really only two cases to consider:\newline

\noindent \textbf{1)}$\tau \models \mapbool(\result, true)$: In this case, rule 6a) is the only one applicable
and $\langle s, \tau \rangle$ will be evaluated next. From $\{p' \wedge \mapbool(\result,true)\} s \{p\}$ we can deduce that the
resulting state $\sigma'$ will again satisfy $p$. We are hence again in a configuration
$\langle \PL{while } e \PL{ do } s \PL{ done}, \sigma' \rangle$ with $\sigma' \models p$. When regarding
$\sigma'$ as equivalent to $\sigma$, then this configuration is equivalent to
the one before applying rule 9. Now this loop in the (abstract) transition system raises the possibility of 
divergence. However, for partial correctness we may disregard this possibility, as we only provide guarantees
for finite computations. The case of divergence will be discussed below.

\noindent \textbf{2)}$\tau \models \mapbool(\result, false)$: In this case, rule 6b) is the only one applicable
and $\langle \PL{null}, \tau \rangle$ is the only statement left to evaluate.
Applying rule 1 leaves us in a final configuration
$\langle \emptyP, \tau' \rangle$ with $\tau' \models p'[\result/b] \wedge \mapbool(b,false) \wedge \result = null$.
As this is the only way for our program to terminate normally,
$\models \{p\} \PL{while } e \PL{ do } s \PL{ done} \{p'[\result/b] \wedge \mapbool(b,false) \wedge \result = null\}$
holds.\qed

\noindent \textbf{Termination:}
For partial correctness, the premise $\{p' \wedge \mapbool(\result, true)\} s;e \{p'\}$ can be
derived from the two other premises by an application of the SEQ rule. It hence does not strengthen the
premises in any way. However, for total correctness, it requires an additional predicate $r(z)$ with $z$
among its free variables, such that
$\forall \sigma \bullet \sigma \models p' \rightarrow \exists z':\typeALnum \bullet r(z')$ and
$\{p' \wedge \mapbool(\result, true) \wedge r(z)\} s;e \{p' \wedge \forall z':\typeALnum \bullet r(z') \rightarrow z' < z\}$
hold. $r(z)$ may be understood as mapping states to sets of natural number values for $z$.
The first requirement thus ensures that the ``mapping'' $r(z)$ is (conditionally) total on all states
in $\llbracket p' \rrbracket$,
while the second requires the loop body $s$ together with the condition $e$ to decrease its supremum.
Since the state $\tau$ reached after evaluating $e$ the
first time satisfied $p'$, by the conditional totality of $r(z)$ we deduce that there must be
an ``initial'' non-empty set of natural numbers $Z$ such that for all $z_i \in Z$, $\tau \models r(z_i)$ holds.
Let $z_{max}$ be the supremum of $Z$. Then, since $z_{max}$ is a natural number and since the supremum
of $Z$ is required to strictly decrease on each loop iteration, there must be a finite number of iterations
after which $z_{max} = 0$. Since there is no natural number smaller than zero, there is no way by which
the second requirement for $r(z)$ can be satisfied on the next iteration. Consequently, the loop has to
terminate after finitely many iterations. \qed
\noindent \textbf{Failsafety:}
A failure might occur either in evaluating $e$ or $s$ or by rule 7 when $e$
evaluates to $\Vnull$. Requiring $e$ and $s$ to both be $\mtagfailsafe$ as well as
$\{p\} e \{\result \not= \Vnull\}$ hence covers all these cases. \qed
\noindent \textbf{Typesafety:}
Same argument as for failsafety applies here, only with the requirement
$\result \not= \Vnull \rightarrow \mapbool(\result)$ instead of $\result \not= \Vnull$.
Note that the case for failure is intentionally left open as typesafe partial correctness only
needs to guarantee the absense of type errors and too strong a premise would lead to incompleteness. \qed
\else
Soundness follows from a standard inductive argument. The extended version \cite{extended_version}
further elaborates the case of the LOOP rule.
\fi

\section{Completeness}
\label{sec:completeness}
In this Section, we will prove the axiomatic semantics of \langDyn{} (relative) complete \cite{Cook78} with
respect to its operational semantics following the seminal completeness proof
of Cook and Gorelick \cite{Cook78,Gorelick75} as well as its extension to OO-programs
due to de Boer and Pierik \cite{DeBoerPierik2003}. That is, given a closed program $\pi$ with a finite set of
class definitions, we prove that $\vDash \{p\} \pi \{q\}$ implies $\vdash_{\PSHL, \PSAL} \{p\} \pi \{q\}$
assuming a complete proof system $\PSAL$ for the assertion language \langAL.

Traditionally, completeness proofs are structured into 3 steps. First, the assertion language is shown
to be expressive, then the system is proven complete for all statements of the programming language and finally,
it is shown to be complete for recursive methods using the concept of most general correctness formulas. Since
both the first and the last step rely on techniques for ``freezing'' program states and for evaluating assertions
on such frozen states,
%and since de Boer and Pierik \cite{DeBoerPierik2003} showed that in object-oriented languages
%these techniques require a much higher degree of sophistication
we follow \cite{DeBoerPierik2003} in prepending a step
for developing adequate freezing techniques for \langDyn.

Completeness proofs for Hoare Logics have been extended and refined for several decades now. Unfortunately,
due to space restrictions we will not be able to give a proper account to the numerous ideas and intriguing
details in the works of our predecessors, but must assume a certain familiarity with such proofs on the
side of the reader.
For the same reason, we will not be able to present the proof as a whole, but will concentrate on those parts 
we had to adapt. % to make it fit to our language and proof system.

\subsection{Freezing the Initial State}

As noticed by Gorelick \cite{Gorelick75}, achieving completeness requires that the assertion language is able to
capture every aspect of a program state in logical variables, in order to ``freeze'' this information during
program execution
and allow the postcondition to compare the initial- to the final state. Pierik and de Boer
\cite{DeBoerPierik2003} pointed out that in OO-contexts this additionally
requires freezing the internal states of all objects existing in the state, necessitating a more sophisticated
freezing-strategy.

While their approach stores objects and the values of their instance variables class-wise, which is difficult
in a dynamically-typed language like \langDyn, the basic idea is fortunately still applicable. 
We use a logical variable $obj$ of type $\typeALobj^*$ to store a (finite) sequence of all existing objects:

\vspace{0.1cm}
\hfill $all(obj) \equiv \forall o: \typeALobj \bullet \exists i: \typeALnum \bullet i < |obj| \wedge obj[i] = o$ \hfill
\vspace{0.1cm}

Since $obj$ establishes a bijection from natural numbers to objects, its allows encoding
states as sequences of natural numbers.
For convenience, we introduce a polymorphic\footnote{We use the polymorphic version for the sake of
readability although the type system of \langAL{} does not allow polymorphism.
However, polymorphic functions can be emulated using one version for each element type.} $pos$ function satisfying
$ \forall \tau : \BaseTypeX \bullet \forall e:\tau, s: \tau^* \bullet s[pos(s, e)] = e $

%First, we define (finite) sequences $inst_{@v}$ of type $\typeALnum^*$ for all (finitely many)
%instance variables $@v$ used in the program and require for each such $@v$:
%
%\[ \forall i:\typeALnum \bullet 0 \le i < |obj| \rightarrow obj[i].@v = obj[inst_{@v}[i]] \]
%
We introduce an enumeration $ivar: \IVar^*$ of all instance variables and
define the following predicate for freezing states:

\vspace{0.1cm}

$code(\overline{x}, obj, \varsigma) \equiv |\varsigma| = |ivar| + 1 \wedge |\varsigma[0]| = |\overline{x}| \wedge obj[0] = \uninitialized \wedge $

\hspace{2cm}$\forall i:\typeALnum \bullet i < |\overline{x}| \rightarrow \varsigma[0][i] = pos(obj, x_i) \wedge$

\hspace{2cm}$\forall i,j: \typeALnum \bullet (i < |ivar| \wedge j < |obj|) \rightarrow ivar[i] = @v \wedge obj[j] = o \rightarrow \varsigma[i+1][j] = pos(obj, o.@v)$

\vspace{0.1cm}

\noindent
where $\overline{x} = x_1,...,x_n$ is a sequence of local variables. The predicate $code(\overline{x}, obj, \varsigma)$ uses
the sequence $obj$ to capture the state of all local variables in $\overline{x}$ as well as all objects in $obj$ in the frozen
state $\varsigma$ of type $(\typeALnum^*)^*$.
%We will furthermore continue to use the symbol $\varsigma$ to denote frozen states.
Note that $\varsigma$ can capture the internal states of all existing objects without referencing any of them.
%(which is important for handling object creation later on).
%
%\[ \forall n:\typeALnum \bullet ivar[n] = @v \leftrightarrow inst[n] = inst_{@v} \]

Also note that this is indeed satisfiable for all states as $\uninitialized \in \typeALobj$ and $\uninitialized \in obj$.
Furthermore, we say that $\varsigma$ encodes $\sigma$ and write
\[\sigma \sim \varsigma \text{ iff } \sigma \models \exists obj: \typeALobj^* \bullet all(obj) \wedge code(\overline{x},obj,\varsigma) \]
with $\{ \overline{x} \} = \LVar \cup \SVar$.

\begin{lemma}[Left-Totality of $\sim$]\label{lem:sim_lt}
 $\forall \sigma: \Sigma \bullet \exists \varsigma: (\typeALnum^*)^* \bullet \sigma \sim \varsigma$.
\end{lemma}

Finally, we are ready to define a predicate transformer $\Theta$ (called the ``freezing function'' in \cite{DeBoerPierik2003}).
However, while in their work, $\Theta$ also bounds all quantification and replaces instance variable dereferencing
by lookups in sequences, we additionally translate all object expressions into expressions of type $\typeALnum$
to allow simulating computations directly on the frozen states.

We hence have the following main cases for our predicate transformer $\Theta^{\overline{x}}_{obj}(\varsigma)$:
\begin{itemize}
 \item $(l.@v)\Theta^{\overline{x}}_{obj}(\varsigma) \equiv \varsigma[pos(ivar,@v)+1][l\Theta^{\overline{x}}_{obj}(\varsigma)]$
 \item $\PL{u}\Theta^{\overline{x}}_{obj}(\varsigma) \equiv \varsigma[0][pos(\overline{x}, \PL{u})]$ \info{where $\PL{u}$ is a program variable in $\overline{x}$}
 \item $u\Theta^{\overline{x}}_{obj}(\varsigma) \equiv u'$ \info{where $u$ is a logical variable of type $\typeALobj$ and $u'$ is a fresh logical variable of type $\typeALnum$}
 \item $(l_1 = l_2)\Theta^{\overline{x}}_{obj}(\varsigma) \equiv l_1\Theta^{\overline{x}}_{obj}(\varsigma) = l_2\Theta^{\overline{x}}_{obj}(\varsigma)$ \info{where $l_1$ and $l_2$ are of type $\typeALobj$.}
 \item $(\exists o:\typeALobj \bullet p)\Theta^{\overline{x}}_{obj}(\varsigma) \equiv (\exists o':\typeALnum \bullet 0 \le o' < |obj| \rightarrow p\Theta^{\overline{x}}_{obj}(\varsigma))$
\end{itemize}

$\Theta^{\overline{x}}_{obj}(\varsigma)$ transforms any assertion $p$ in such a way that it operates
on the frozen state $\varsigma$ instead of the real program variables. Like the $\Theta$ in \cite{DeBoerPierik2003}, our
$\Theta^{\overline{x}}_{obj}(\varsigma)$ hence satisfies the following property

\begin{theorem}[Invariance]\label{thm:freezing}
 $\vdash \{p\Theta^{\overline{x}}_{obj}(\varsigma)\} s \{p\Theta^{\overline{x}}_{obj}(\varsigma)\}$
for all statements $s$ and assertions $p$ as long as $\overline{x}$ contains all program variables used and
$obj$ contains all objects accessed in $p$.
\end{theorem}

It can hence replace $\Theta$ in the remaining argument. Note that
$p\Theta^{\overline{x}}_{obj}(\varsigma)$ is a property of $\varsigma$ as its truth value is independent of
any particular state. We hence write $\models p\Theta^{\overline{x}}_{obj}(\varsigma)$ if its truth value is true.
Also observe% the following property of $\Theta^{\overline{x}}_{obj}(\varsigma)$:

\begin{lemma}[Freezing]\label{lem:theta}
 $\sigma \models q \hspace{0.5cm}\text{ iff }\hspace{0.5cm} \sigma \sim \varsigma \wedge \models q\Theta^{\overline{x}}_{obj}(\varsigma)$
\end{lemma}
\ifx \version\extended
\begin{proof}
 By induction over the structure of $q$.
\end{proof}
\fi

\subsection{Expressivity}

Cook \cite{Cook78} first discussed the importance of an expressive assertion language for the completeness of
a Hoare logic. In essence, the assertion language must be able to express the strongest postcondition $SP(s,p)$
for all statements $s$ and preconditions $p$.

In the last Section, we already established that it is possible to capture all information about a state in
a structure consisting of finite sequences of natural numbers. Using G\"odelization,
one can take this a step further and encode these sequences themselves as a single natural number. Then,
we consider a predicate $comp_s$ of type $\typeALnum \times \typeALnum \mapsto \typeALnum$ simulating \langDyn{}
computations on such frozen states and note that,
since such computations are by definition computable, it can be defined as a $\mu$-recursive function.

By
\ifx \version\extended
Theorem~\ref{thm:mu_rec_in_ass}
\else
\cite[Theorem 6]{extended_version}
\fi, it is hence expressible in our assertion language and we can use it
within our assertions without any loss of generality. For convenience, we will omit the G\"odelization step
and instead use a version of $comp_s$ operating on frozen states as defined above.
To formalize the idea that $comp_s$ simulates \langDyn{} computations on frozen states, we stipulate

\begin{lemma}\label{lem:comp_s}
 $comp_s = \{(\varsigma, \varsigma') \mid \forall \sigma, \sigma' \bullet (
      \sigma \sim \varsigma \wedge \sigma' \sim \varsigma') \rightarrow \sigma' \in \semantics\llbracket s \rrbracket(\sigma)
      \}$
\end{lemma}

Using $comp_s$ we can show the following:

\begin{theorem}[Definability of Weakest Preconditions]\label{thm:def_wp}
 For all postconditions $q$ and statements $s$, the precondition

\vspace{0.1cm}
 $p \equiv \forall \varsigma, \varsigma' \bullet (all(obj) \wedge code(\overline{x}, obj, \varsigma') \wedge comp_s(\varsigma', \varsigma)) \rightarrow q\Theta^{\overline{x}}_{obj}(\varsigma)$
\vspace{0.1cm}

 satisfies $\llbracket p \rrbracket = \{ \sigma \mid \semantics\llbracket s \rrbracket(\sigma) \subseteq \llbracket q \rrbracket \}$.
\end{theorem}
\todo{do we need to also introduce predicates typeerror(), fail() and diverge()?}

\ifx \version\extended
The proof can be found in appendix~\ref{app:omitted_proofs}.
\else
The proof can be found in \cite{extended_version}.
\fi
Since definability of weakest preconditions is equivalent to the definability of strongest postconditions
\cite{Olderog83}, we have

\begin{theorem}[Expressiveness]
 The assertion language \langAL{} is expressive with respect to its standard interpretation and the programming language \langDyn{}.
\end{theorem}

\subsection{Completeness for Statements}\label{sec:comp_stmt}

As usual \cite{Cook78,Gorelick75}, the core of our completeness proof consists of an induction over the
structure of a statement $s$.
Since several of our rules deviate from theirs, we need to exchange these
cases in argument. We will concentrate on the most interesting cases.

\smallskip

\noindent\textbf{Induction Basis:}
\begin{itemize}
\ifx \version\extended
 \item $s \equiv \Vnull$: Assume $\models \{p\} \PL{null} \{q\}$. Then, by the operational semantics,
 $p \rightarrow q[\result := \Vnull]$ must also be true. It is hence derivable in $\PSAL$ and the desired
 result follows from the CONST axiom followed by applying the rule of consequence (CONS).
 \textit{Typesafety:} The CONST axiom always derives typesafety. Should $\mtagtypesafe$ not be required, it can be omitted
 in the rule of consequence. The same holds for $\mtagfailsafe$ and $\mtagterminates$.
\fi
 \item $s \equiv \PL{u}$: Assume $\models \{p\} \PL{u} \{q\}$. Then, by the operational semantics,
 $p \rightarrow q[\result := \PL{u}]$ must also be true. It is hence derivable in $\PSAL$ and the desired
 result follows from the VAR axiom followed by applying the rule of consequence (CONS).
 \textit{Typesafety:} Assume $\models \{p\} \PL{u} \{\mtagtypesafe \wedge q \}$. Then $\{p\} \PL{u} \{q\}$
 and $\{p\} \PL{u} \{\mtagtypesafe\}$ must also be true. The former can thus be derived using above argumentation
 and the latter implies $p \rightarrow \PL{u} \not= \uninitialized$, which is hence derivable in $\PSAL$ and
 the axiom VAR-TAG followed by an applying the rule of consequence (CONS) derives $\{p\} \PL{u} \{\mtagtypesafe\}$.
 Now the rule of conjunction (CONJ) followed by the rule of consequence (CONS) derives the desired result.
 $\mtagfailsafe$ and $\mtagterminates$ can be derived using the axiom VAR-TAG without any preconditions.
 \item $s \equiv \PL{@v}$: Just like the case for $\PL{u}$, applying IVAR instead of VAR and IVAR-TAG instead of VAR-TAG.
\end{itemize}

\noindent\textbf{Induction Hypothesis:} $\models \{p\} s \{q\} \rightarrow \vdash_{\PSHL,\PSAL} \{p\} s \{q\}$ for
all assertions $p,q$ and all statements $s$ of a program $\pi$ containing no recursive method calls.

\smallskip

\noindent\textbf{Induction Step:}
\begin{itemize}
 \item $\PL{u} := e$: Assume $\models \{p\} \PL{u} := e \{\mtagvarX \wedge q\}$. Then according to the operational semantics,
  $\{p\} e \{\mtagvarX \wedge q[\PL{u} := \result]\}$ must also be. By the induction hypothesis, it is hence derivable.
  An application of the rule ASGN derives the desired result.
\ifx \version\extended
 \item $s \equiv s_1;s_2$: Assume $\models \{p\} s_1;s_2 \{\mtagvarX \wedge q\}$. Then by the expressibility of
 the assertion language, there is an intermediate assertion $r$ such that $\{p\} s_1 \{\mtagvarX \wedge r\}$
 and $\{r\} s_2 \{\mtagvarX \wedge q\}$ are
 also true. Hence by the induction hypothesis both are derivable and an application of the rule SEQ derives the
 desired result.
\fi
 \item $s \equiv \PL{\textbf{if} }e\PL{ \textbf{then} }s_1\PL{ \textbf{else} }s_2\PL{ \textbf{end}}$:  Assume $\{p\} \PL{\textbf{if} }e\PL{ \textbf{then} }s_1\PL{ \textbf{else} }s_2\PL{ \textbf{end}} \{\mtagvarX \wedge q\}$ is true.
 Then, by the expressiveness of the assertion language and the operational semantics, there is an intermediate
 assertion $r$ such that
 $\{p\} e \{\mtagvarX \wedge r\}$, $\{r \wedge \mapbool(\result, true) \} s_1 \{ \mtagvarX \wedge q \}$ and 
 $\{r \wedge \mapbool(\result, false) \} s_2 \{ \mtagvarX \wedge q \}$ are also true and hence derivable
 by the induction hypothesis. Now an application of the rule COND derives the desired result.
 \textit{Failsafety:} Since above argumentation already ensured that $e$, $s_1$ and $s_2$ are all failsafe, the
 only additional requirement is $\{p\} e \{\result \not= \Vnull\}$. However, since the case $\result = \Vnull$
 leads to failure in the operational semantics, this must hold for any execution of $s$ in order to be failsafe and
 hence must be derivable by the induction hypothesis.
 
 \textit{Typesafety:} The same argumentation as for failsafety applies here, only the additional requirement is
 $\{p\} e \{\result \not= \Vnull \rightarrow \mapbool(\result)\}$. Note that the case of $\result = \Vnull$ can be
 deliberately allowed, since it leads to a failure in the operational semantics and thus does not affect typesafety.
 \item $s \equiv \PL{\textbf{while} } e \PL{ \textbf{do} } s_1 \PL{ \textbf{done}}$: Assume $\{p\} \PL{\textbf{while} } e \PL{ \textbf{do} } s_1 \PL{ \textbf{done}} \{\mtagvarX \wedge q\}$ is true.
 Then, by the standard argument for while loops due to Cook \cite{Cook78} (and explained particularly well by Apt \cite{Apt81}),
 the expressiveness of the assertion language and the operational semantics, there are two assertions $i$ and $i'$ such that
 $p \rightarrow i$, $\{i\} e \{\mtagvarX \wedge i'\}$, $\{i' \wedge \mapbool(\result, true)\} s_1 \{\mtagvarX \wedge i\}$ and
 $i'[b/\result] \wedge \mapbool(b, false) \wedge \result = \Vnull \rightarrow q$ are true and hence
 derivable by the induction hypothesis and the completeness of $\PSAL$. While $i$ is the loop invariant of $s$, $i'$ is an intermediate state neccessary because in \langDyn, $e$ could
 have side-effects. Now, an application of the LOOP rule followed by the rule of consequence derives the
 desired result.
 \ifx \version\extended
 
 \textit{Termination:}
 Assuming $\{p\} \PL{\textbf{while} } e \PL{ \textbf{do} } s_1 \PL{ \textbf{done}} \{\mtagterminates \wedge q\}$, then there is a
 $\mu$-recursive function $v(\varsigma)$ simulating the execution of $s$ using $comp_s$ and determining
 the least number of iterations it takes to reach a state $\tau$ from the current state such that $e$
 evaluates to false in $\tau$. Note that by Theorem~\ref{thm:mu_rec_in_ass} $v(\varsigma)$ can expressed
 in \langAL(). Also, by our assumption that the loop $s$ terminates, the function $v(\varsigma)$
 is well-defined on all states in $p'$ and thus $r(z) \equiv all(obj) \wedge code(\overline{x},obj,\varsigma) \wedge z = v(\varsigma)$
 is a canonical loop variant satisfying
 $\forall \sigma \bullet \sigma \models p' \rightarrow \exists z':\typeALnum \bullet r(z')$.
 Since $v(\varsigma)$ determines the number of iterations until reaching a target state, executing $s;e$
 clearly decreases it and thus
 $\models \{p' \wedge \mapbool(\result, true) \wedge r(z)\} s;e \{p' \wedge \forall z':\typeALnum \bullet r(z') \rightarrow z' < z\}$
 holds. By the induction hypothesis, it is thus derivable. An application of the LOOP rule derives the
 desired result.
 \else
  \textit{Termination:} see \cite{extended_version}
  %\footnote{The argument for termination is included in \cite{extended_version}}
 \fi
 
 \textit{Failsafety \& Typesafety:} the exact same argument as for conditionals applies here as well.
 \todo{constructor call (maybe add below (after method call))}
\end{itemize}

\subsection{Completeness for Recursive Methods}

The methodology for proving a Hoare logic complete for recursive procedures by using most general correctness
formulas is due to Gorelick \cite{Gorelick75}. It was extended to
OO-programs by De Boer and Pierik \cite{DeBoerPierik2003}.

A curious implication of dynamic dispatch under dynamic typing is that the lack of type information
prohibits pinpointing the exact target of a method call. For instance, the weakest precondition of the call
$\PL{x.size}()$ with respect to the postcondition $\mapnum(\result, 5)$ must include all possibilities
like the case of the variable $x$ referring to a string of length $5$ as well as $x$ referring to a list
of size $5$.
In general, the weakest precondition of a method call $l.\PL{m}(v_1,...,v_n)$ is the disjunction of all
weakest preconditions derivable as described in the proof of Theorem~\ref{thm:MGF} from the most
general correctness formulas of all methods $\PL{C.m}$ of arity $n$ of all classes $C \in \Class$, each
conjoined with the corresponding type assumption $\llbracket l \rrbracket \in \{C\}$.
Note that this methodology introduces an implicit closed world assumption as it fails
when using a method with a different set of classes. However, we regard this
problem as one of modularity rather than completeness and thus out of scope.

% Dynamic Dispatch under Dynamic Typing: In Object-Oriented languages, the method called in a method call
%  depends on the (runtime) type of its receiver. Contrary to static typing, dynamic typing usually does not
%  provide enough information to pinpoint it statically. The weakest precondition for a method call
%  $l.\PL{m}(v_1,...,v_n)$ hence must be a disjunction over all possibilities: the weakest preconditions calculated
%  as described under 1) and 2) from the most general correctness formulas of all methods $\PL{m}$ with arity $n$ of all
%  classes in $\pi$.\todo{move someplace else. Discussion of the WP-calculus, maybe?}
%
As our tagged Hoare logic incorporates different notions of correctness, we generalize Gorelick's idea
to a set of most general correctness formulas.
The most general correctness formulas for a statement $s$ are

\noindent $MGF(s) = \{ \{WP(s,init)\} s \{init\} \} \cup \{ \{WP_{\mtagtagsmall}(s,true)\} s \{\mtagtag\} \mid \mtagtag \in \mtags \}$

\info{with $init \equiv all(obj) \wedge code(\overline{x},obj,\varsigma)$}. The reason for this
is obvious: From $MGF(s)$, we can deduce $\{WP_{\mtagvarXsmall}(s, q)\} s \{\mtagvarX \wedge q\}$
with $\mtagvarX \subseteq \tags$ using the conjunction rule. The converse is not in all cases possible.

The results from Section~\ref{sec:comp_stmt} imply that above set 
can be derived for any \langDyn{}
statement $s$ given that they are true. 
Should, e.g.,  $s$ raise a type error on all inputs then
$WP_{\mtagtypesafesmall}(s,true) \equiv false$ and $\{false\} s \{\mtagtypesafe\}$ is derivable.

\begin{theorem}[MGFs]\label{thm:MGF}
  $\models \{p\} s \{\mtagvarX \wedge q\}  \rightarrow  MGF(s) \vdash_{\PSHL,\PSAL} \{p\} s \{\mtagvarX \wedge q\}$
\end{theorem}

\begin{proof}
 Assume $\models \{p\} s \{\mtagvarX \wedge q\}$. Then $\{p\} s \{q\}$ and $\{p\} s \{\mtagtag\}$ for all
 $\mtagtag \in \mtagvarX$ are also all true.\newline
 \noindent\textbf{1)} $\vdash \{p\} s \{q\}$:
 For technical convenience only we assume that $p$ and $q$ do not contain free occurrences of the logical variables used to
 freeze states. If they do, these need to be renamed using the substitution rule.
 By Theorem~\ref{thm:freezing} we have $\{q\Theta^{\overline{x}}_{obj}\} s \{q\Theta^{\overline{x}}_{obj}\}$. An
 application of the conjunction rule yields

 \vspace{0.05cm}
 \hfill $ \{q\Theta^{\overline{x}}_{obj} \wedge WP(s,init)\} s \{q\Theta^{\overline{x}}_{obj} \wedge init\} $ \hfill
 \vspace{0.05cm}

 Next, we have to prove $p \rightarrow q\Theta^{\overline{x}}_{obj} \wedge WP(s,init)$.
 Assume $\sigma \models p$. Then by $\models \{p\} s \{q\}$, for all $\sigma' \in \semantics\llbracket s \rrbracket(\sigma)$,
 we have $\sigma' \models q$. By Lemma~\ref{lem:theta}, we have $\sigma' \models q\Theta^{\overline{x}}_{obj} \wedge init$.
 Now, by Theorem~\ref{thm:freezing}, we have $\vdash \{q\Theta^{\overline{x}}_{obj}\} s \{q\Theta^{\overline{x}}_{obj}\}$,
 and by soundness of our proof system $\models \{q\Theta^{\overline{x}}_{obj}\} s \{q\Theta^{\overline{x}}_{obj}\}$.
 Hence, $\sigma \models q\Theta^{\overline{x}}_{obj}$ and by the definition of $WP$, $\sigma \models WP(s, init)$.
 Therefore, $p \rightarrow q\Theta^{\overline{x}}_{obj} \wedge WP(s, init)$ holds and since $q\Theta^{\overline{x}}_{obj} \wedge init \rightarrow q$
 follows directly from Lemma~\ref{lem:theta}, an application of the rule of consequence derives $\{p\} s \{q\}$.
 \newline\noindent\textbf{2)} $\vdash \{p\} s \{\mtagtag\}$:
 if true, then $p \rightarrow WP_{\mtagtagsmall}(s,true)$ must also be and is hence derivable by the
 completeness of $\PSAL$. Since $\{WP_{\mtagtagsmall}(s,true)\} s \{\mtagtag\} \in MGF(s)$, an application of
 the consequence rule derives the desired result.
 \newline\noindent\textbf{3)} $\vdash \{p\} s \{\mtagvarX \wedge q\}$: One application of the conjunction rule per tag
 in $\mtagvarX$ completes the proof.\qed
\end{proof}

Finally, since our recursion rule is identical to the one devised by Gorelick \cite{Gorelick75} for this purpose,
we are now able to apply the same inductive argument used by Gorelick for proving our Hoare logic complete for recursive
methods.

\begin{lemma}\label{lem:comp_rec_meth}
Let $M_i \equiv l_i.m_i(\Vector{v_i})$ denote the $i$th (possibly recursive) method call
occurring in a closed program $\pi$
and let $A = \bigcup^{n}_{i = 1} MGF(M_i)$ be the set of most general correctness formulas about these method calls
then for all statements $s$ of $\pi$ and all assertions $p$ and $q$: 
$ \models \{p\} s \{q\} \rightarrow A \vdash_{\PSHL,\PSAL} \{p\} s \{q\}$
\end{lemma}

\begin{proof}
 By induction over the structure of $s$. Most cases are as in the proof for the non-recursive case.
 Most interesting is the new case for method calls:
 %\begin{itemize}
  %\item 
  $s \equiv l_i.m_i(\Vector{v_i})$: Assuming $\models \{p\} s \{q\}$ and $s$
  is the $i$th method call $M_i$ in our program, then $MGF(s) \subseteq A$ and hence
  $A \vdash \{p\} s \{q\}$ by Theorem~\ref{thm:MGF}. As Gorelick \cite{Gorelick75} pointed out,
  this also holds for recursive method calls.
 %\end{itemize}
\end{proof}
\todo{with Tags}

\begin{theorem}[Completeness for Recursive Methods]\label{thm:comp_rec_meth}\newline
 \vspace{0.1cm}
 \hfill $\models \{p\} s \{\mtagvarX \wedge q\} \hspace{1cm} \rightarrow \hspace{1cm} \vdash_{\PSHL,\PSAL} \{p\} s \{\mtagvarX \wedge q\}$ \hfill
 \vspace{0.1cm}

 \noindent for any statement $s$ of a closed program $\pi$ containing possibly recursive method calls and all
 assertions $p$ and $q$.
\end{theorem}

\begin{proof}
 Expressiveness of \langAL{} guarantees the expressibility of $WP_{\mtagvarXsmall}(s,q)$ for any statement
 $s$ and postcondition $q$. Hence by setting $q \equiv init$ and $s \equiv M_i$ for any
 $i \in \mathbb{N}^1_n$ we can see that the set $A$ of most general
 correctness formulas of all method calls is expressible in our logic. Now, since by definition of $WP_{\mtagvarXsmall}$,
 these formulas are true, we have by Lemma~\ref{lem:comp_rec_meth}

 $A \vdash_{\PSHL,\PSAL} \{p_i\} \PL{\textbf{begin local} } \PL{self},\Vector{u_i} := l_i, \Vector{v_i}; s_i \PL{ \textbf{end}} \{q_i\}$ as well as

 $A \vdash_{\PSHL,\PSAL} \{p_{\mtagtagsmall,i}\} \PL{\textbf{begin local} } \PL{self},\Vector{u_i} := l_i, \Vector{v_i}; s_i \PL{ \textbf{end}} \{\mtagtag\} \text{ for all $\mtagtag \in \mathcal{T}ags$} $

 with $p_i \equiv WP(M_i,init)$, $q_i \equiv init$, $p_{\mtagtagsmall,i} \equiv WP_{\mtagtagsmall}(M_i,true)$ and $s_i$ denoting the method body of the method called in $M_i$
 for all $i \in \mathbb{N}^1_n$. Note that above statements establish the assumptions in the set $A$ and
 together allow deriving the assumptions for the REC rule of the form

% \vspace{0.05cm}
$ A \vdash_{\PSHL,\PSAL} \{p_i \wedge r_i(z)\} \PL{\textbf{begin local} } \PL{self},\Vector{u_i} := l_i, \Vector{v_i}; s_i \PL{ \textbf{end}} \{\mtagvarX_i \wedge q_i\} $
% \vspace{0.05cm}

 for all $i \in \mathbb{N}^1_n$. As for the case not concerned with termination, we can simply set $r_i(z) \equiv z=z$.
 Furthermore, assuming $\models \{p\} s \{q\}$, by Lemma~\ref{lem:comp_rec_meth} we have

% \vspace{0.05cm}
$ A \vdash \{p\} s \{\mtagvarX \wedge q\} $
% \vspace{0.05cm}

 Now these are just the premises of the REC rule. Note that in the case not concerned with termination,
 the set of assumptions $A$ is derivable from $A'$ by applying the consequence rule to each element.
 Hence, an application of the REC rule derives the desired result and completes the proof.
\todo{dynamic dispatch?} 
 \newline\noindent\textbf{Termination:} for proving termination of \langDyn{} programs, the rules LOOP and REC
 must be altered to support so-called loop-variants or recursion bounds. Usually, these take the form
 of an integer expression $t$ whose value a) must be $> 0$ whenever the loop / recursive method is entered
 (thus forcing termination when reaching zero) and b) must decrease on every iteration / recursive call.
 Note that this methodology syntactically restricts the loop variant / recursion bound to be an integer
 expression of the assertion language. Now, as observed by Apt, De Boer and Olderog in
 \cite{AptDeBoerOlderogBook2009}\todo{more accurate}, this method introduces incompleteness in the case of
 total correctness, since it assumes the integer expressions of the assertion language to be able to express
 any necessary loop-variant / recursion bound.
 However, while-loops and recursive methods allow \langDyn{}-programs to calculate any
 $\mu$-recursive function and hence obviously also to bound the number of loop iterations by any
 $\mu$-recursive function, while the set of integer operations available in the assertion language might
 be quite limited (e.g. in our case lacking exponentiation).
 We circumvent this problem by introducing a new form of loop-variants and recursion bounds, which allow the
 use of quantifiers. The old form used a logical variable $z$ of type $\typeALnum$ to store the value of $t$
 before a loop iteration ($t = z$ in the precondition) and compare it to the new value in the postcondition
 ($t < z$).
 Our new form uses a predicate $r(z)$ with $z$ among its free variables instead of $t = z$ and the
 logical expression $\forall z':\typeALnum \bullet r(z') \rightarrow z' < z$ where $r(z')$ denotes the result
 of substituting $z'$ for $z$ in $r$ instead of $t < z$.
 Firstly, observe that this is a conservative extension as one may set $r \equiv t = z$ for some integer
 expression $t$.
 Secondly, note that by
 \ifx \version\extended
 Lemma~\ref{lem:trans_mu_rec_to_ass}
 \else
 \cite[Lemma 5]{extended_version}
 \fi, $r$ may compute any $\mu$-recursive function
 and is thus contrary to integer expressions able to express any function computable by \langDyn{}-programs
 including exponentiation.\qed
\end{proof}\todo{is this sufficient to proof completeness for total correctness?}

\ifx \version\extended
\section{The Translational Approach}
\label{sec:translational_approach}
The translational approach was introduced by Apt, De Boer and Olderog in \cite{AptDeBoerOlderog2012}
to facilitate the availability of sound and complete axiomatic proof systems for different programming languages.
The basic idea is to transfer soundness and completeness results for their proof systems from language to language
by means of a (more intuitive and usually much simpler) semantics-preserving translation between the programming
languages. The program logic presented in this work handles the fundamental issues of dynamically typed languages
and hence opens the gate for using the translational approach to prove program logics for
such programming languages sound and complete in future.
In the following, we will list some ideas on how more advanced dynamic features found in real-world
dynamically-typed languages like JavaScript, Ruby and Python might be translated to \langDyn{}.

\subsection{Method Update}

Languages combining class-based object orientation with dynamic typing (like Ruby and Python) often support
a feature we call ``method update'' allowing programs to override methods at runtime - most often it is not
even required for the arity of the new method to match the old version.

\noindent\textbf{Translation:} First, for each method $C.m$ in the original program having multiple versions,
the corresponding \langDyn{} program must have a global state $g_{C.m}$ for storing the information which
version is the current one. Since \langDyn{} program usually do not have any global state, we accomplish this
by introducing a class $Global$ encapsulating this state information and passing a reference $g$ to its only
instance into each and every method in the program.
Second, let there be versions $C.m_1, ..., C.m_k$ of a method $C.m$
with arities $n_1,...,n_k$ within the original program, then the function $a_{C.m}(n) = \{ C.m_i \mid n_i = n \}$
groups all versions having the same arity. For each arity $n$, such that $|a_{C.m}(n)| > 0$, the corresponding
\langDyn{} program contains a method $C.m_n$ with arity $n$ whose body is structured as follows

\textbf{if}\verb| g.version_C_m() == 1| \textbf{then}

\verb| # body of C.m_1 here|

\verb|else|

\verb|  if g.version_C_m() == 2 then|

\verb|   # body of C.m_2 here|

\verb|  else|

\verb|  ...|

\verb|        else|

\verb|          typeerror|

\verb|        end|

\verb|  end|

\verb|end|

where $a_{C.m}(n) = \{C.m_1,C.m_2,...,C.m_l\}$.

Then, we only have to treat the updating of a method $C.m$ as setting its global state to a new value $v$
by \verb|g.set_version_C_m(v)|.

Furthermore, when translating every method call to a method $C.m$ of arity $n$ in the original program as
a call to $C.m_n$ in the corresponding \langDyn{} program, the result should be behaviourally equivalent.

\subsection{Closures}

Another feature of functional languages that is often found in dynamically typed
languages are closures (E.g. JavaScript functions or ruby blocks). Closures are characterized by two properties:
Firstly, they allow passing around (a reference to) code as data and secondly, they capture the values of all
free variables within their body upon creation.

\noindent\textbf{Translation:} In \langDyn{}, we can emulate both properties by introducing a new class
$C_c$ for each closure $c$. This class defines a method $do()$ of the same arity as the closure and whose
method body is just $c$'s body with all free variables replaced by corresponding instance variables as well
as a constructor $init$ taking all variables as arguments that occur free in $c$'s body and storing their
value in corresponding instance variables. Now, the closure definition $c = \lambda p_1,...,p_n. s$ can be
replaced by
\[c = \mathbf{new}\PL{ }C_c(v_1,...,v_k)\]
where $v_1,...,v_k$ are all variables occuring free in $s$ and each call $c(a_1,...,a_n)$ can be replaced
by the method call
\[c.do(a_1,...,a_n)\]
The resulting program should be behaviourally equivalent. Note that a finite program can only contain a finite
number of closures and this replacement hence will only introduce a finite number of additional classes $C_c$.

\subsection{Multiple Return Values}

In Ruby, methods are allowed to have multiple return values. Under the hood, this is realized by converting
the return values into a list and assigning the list elements to their respective variables.

\noindent\textbf{Translation:}
Since \langDyn{} also supports heterogeneous lists, one can create a similar mechanism by translating
\[\PL{return } e_1,...,e_n \text{ \hspace{2cm} into \hspace{2cm}} \PL{return } [e_1,...,e_n] \]
and
\[v_1,...,v_n = C.m() \text{ \hspace{1cm} into \hspace{1cm} } l = C.m(); v_1 = l[0]; ...; v_n = l[n-1] \]
where $l$ is a fresh local variable.

\fi

\section{Conclusions \& Outlook}
\label{sec:conclusion}

We presented a sound and (relative) complete Hoare logic for \langDyn{}.
Open are the issues of modularity (applicability to open programs) and
allowing tags carrying additional information (to incorporate extensions like De Boer's footprints \cite{DeBoerDeGouw2015Invariance}).

% Future Work
% - extend Tagged Hoare Logic to allow for tags to carry additional information -> would allow to express
%   extensions like Frank de Boer's Footprints as tags
% - open programs

\smallskip

\textbf{Acknowledgements:} We thank Dennis Kregel for noticing that restricting $\result$ causes
incompleteness and him, Nils-Erik Flick and the anonymous referees for many useful comments on prior
versions of this paper.
% eval()
% 

\todo{embed the bibliography}
\bibliographystyle{splncs03}
\bibliography{specific}

\ifx \version\extended

\appendix

\section{Omitted Proofs}
\label{app:omitted_proofs}

Notation: We sometimes use $p\frac{v}{t}$ to denote the result of substituting $t$ for $v$ in $p$.

\vspace{0.3cm}

\noindent Proof of Theorem~\ref{thm:def_wp}:
\begin{proof}
 We have to prove the equality $LHS = RHS$ where
 $LHS \equiv \llbracket p \rrbracket$ and $RHS \equiv \{ \sigma \mid \semantics\llbracket s \rrbracket(\sigma) \subseteq \llbracket q \rrbracket\}$.
 We will first prove the direction $LHS \subseteq RHS$ and then turn to the converse question.\newline
 \textbf{1)} $LHS \subseteq RHS$: Assuming a state $\sigma \in LHS$, then by left-totality of $\sim$, there
 is a $\varsigma'$ such that $\sigma \sim \varsigma'$. Furthermore, from
 $\sigma \models \forall \varsigma \bullet comp_s(\varsigma', \varsigma) \rightarrow q\Theta^{\overline{x}}_{obj}(\varsigma)$ and
 Lemma~\ref{lem:comp_s} we can deduce that every state $\sigma' \in \semantics\llbracket s \rrbracket(\sigma)$
 has a $\varsigma$ satisfying $\sigma' \sim \varsigma$ as well as
 $comp_s(\varsigma', \varsigma)$. Since all premises of the implication on the left-hand side are satisfied,
 $\models q\Theta^{\overline{x}}_{obj}(\varsigma)$ must hold as well. Note that the latter two are properties of
 $\varsigma$ and $\varsigma'$ rather than any particular state. Using Lemma~\ref{lem:theta} we can then deduce
 $\sigma' \models q$ and since our only assumption about $\sigma'$ is that it is a post-state of $\sigma$, it follows
 that $\semantics\llbracket s \rrbracket(\sigma) \subseteq \llbracket q \rrbracket$ and hence that $\sigma \in RHS$.

 \textbf{2)} $RHS \subseteq LHS$: Assume $\sigma \in RHS$. $\sigma$ is hence an initial state and all its post-states $\sigma' \in \semantics\llbracket s \rrbracket(\sigma)$
 satisfy the assertion $q$. Then, by left-totality of $\sim$, there is a frozen state $\varsigma'$ such that
 $\sigma \sim \varsigma'$ and for every post-state $\sigma'$ there is a frozen state $\varsigma$ such that
 $\sigma' \sim \varsigma$.
 Now, by Lemma~\ref{lem:comp_s}, every pair of (frozen) pre- and post-state $(\varsigma', \varsigma) \in comp_s$.
 Also, since the post-states $\sigma'$ satisfy $q$ and $\sigma' \sim \varsigma$, by Lemma~\ref{lem:theta} we know that
 $\models q\Theta^{\overline{x}}_{obj}(\varsigma)$ holds.
 Therefore, the entire assertion $p$ is true in $\sigma$ and hence $\sigma \in LHS$. \qed
\end{proof}

\begin{lemma}\label{lem:trans_mu_rec_to_ass}
 For every $\mu$-recursive $k$-ary function $f$, there exists a formula $p$ in $\langAL{}$ with free
 variables $r, x_1,...,x_k$, such that
 \[ f(a_1,...,a_k) = z \text{ iff } \models p \frac{r, x_1, ..., x_k}{z, a_1, ..., a_k} \]
\end{lemma}

\begin{proof}
 By induction over the structure of $\mu$-recursive functions.
 \begin{itemize}
  \item If $f$ is a constant function $f(x_1,...,x_k) = n$, then the formula $p \equiv r = n$ satisfies the Lemma.
  \item If $f$ is the successor function $f(x_1) = x_1 + 1$, then the formula $p \equiv r = x_1 + 1$ satisfied the Lemma.
  \item If $f$ is the projection $f(x_1,...,x_n) = x_i$, then the formula $p \equiv r = x_i$ satisfies the Lemma.
  \item If $f$ is a composition of a $k$-ary function $h$ and $k$ $n$-ary functions $g_1,...,g_k$, then by the
  induction hypothesis, there are formulas $p_h, p_{g_1},...,p_{g_k}$ corresponding to the functions $h, g_1,..., g_k$ as
  described in the Lemma. Then, $p \equiv \exists v_1,...,v_k: \typeALnum \bullet p_h \frac{x_1,...,x_k}{v_1,...,v_k} \wedge p_{g_1}\frac{r}{v_1} \wedge ... \wedge p_{g_k}\frac{r}{v_k}$
  satisfies the Lemma.
  \item If $f$ is a primitive recursion $\rho(g,h)$ with a $k$-ary function $g$ and a $k+2$-ary function $h$, then according to
  the induction hypothesis, there are formulas $p_g$ and $p_h$ corresponding to the functions $g$ and $h$ as described in the Lemma.
  Now, $p \equiv \exists s:(\typeALnum^k)^* \bullet |s| = x_1 \wedge p_{g}\frac{r,x_1,...,x_k}{s[0],x_2,...,x_{k+1}}) \wedge$
  $\forall i : \typeALnum \bullet 0 \le i < |s|-1 \rightarrow p_{h}\frac{r,x_1,...,x_{k+2}}{s[i+1],i,s[i],x_1,...,x_k}$ satisfies the Lemma.
  \item If $f$ is a minimization $\mu f$ of a $k$-ary function $f$, then according to the induction hypothesis, there
  is a formular $p_f$ corresponding to $f$ as described in the Lemma. Now,
  $p \equiv \exists v: \typeALnum \bullet p_f\frac{r,x_1,...,x_k}{0,v,x_1,...x_{k-1}} \wedge \forall v':\typeALnum \bullet \exists v_r: \typeALnum \bullet v' < v \rightarrow p_f\frac{r,x_1,...,x_k}{v_r,v',x_1,...,x_{k-1}} \wedge v_r > 0$
  satisfies the Lemma.
 \end{itemize}

\end{proof}

\begin{theorem}\label{thm:mu_rec_in_ass}
 Every $\mu$-recursive function is expressible in \langAL{}.
\end{theorem}

\begin{proof}
 Direct consequence of Lemma~\ref{lem:trans_mu_rec_to_ass}.
\end{proof}

\section{Substitutions}
\label{app:substitutions}
Analogous to the state update operation $[u := e]$, the program logic uses 3 different kinds of substitutions
on assertions.

\subsection{Substitutions}

\textbf{1. Substitution of local variables $p[x := e]$}

The substitution for local variables (or multiple local variables in parallel)
is straightforward.

It is defined by induction on the structure of $p$:
\begin{itemize}
 \item $y[x := e] \equiv \begin{cases} e \text{ if $x = y$} \\ y \text{ otherwise} \end{cases}$ (includes $y \equiv \self$)
 \item $v[x := e] \equiv v$
 \item $n$, $true$, $false$ (constants - unaffected)
 \item $l.@v[x := e] \equiv l[x := e].@v$
 \item $\PL{\textbf{if} } l \PL{ \textbf{then} } l_1 \PL{ \textbf{else} } l_2 \PL{ \textbf{end}}[x := e] \equiv \PL{\textbf{if} } l[x := e] \PL{ \textbf{then} } l_1[x := e] \PL{ \textbf{else} } l_2[x := e] \PL{  \textbf{end}}$
 \item $l_1 = l_2 \equiv l_1[x := e] = l_2[x := e]$
 \item $l_1 < l_2 \equiv l_1[x := e] < l_2[x := e]$
 \item $\llbracket l \rrbracket \in \{C_1,...,C_n\} \equiv \llbracket l[x := e] \rrbracket \in \{C_1,...,C_n\}$
 \item $l_1 \wedge l_2 \equiv l_1[x := e] \wedge l_2[x := e]$
 \item $l_1 \vee l_2 \equiv l_1[x := e] \vee l_2[x := e]$
 \item $(\neg l_1)[x := e]\equiv \neg l_1[x := e]$
 \item $(\exists y:T. l_1)[x := e] \equiv \begin{cases}
                                           \exists y:T. l_1\text{ if $x = y$} \\
                                           \exists y:T. l_1[x := e]\text{ otherwise}
                                          \end{cases}$
 \item $(\forall y:T. l_1)[x := e] \equiv  \begin{cases}
                                           \forall y:T. l_1\text{ if $x = y$} \\
                                           \forall y:T. l_1[x := e]\text{ otherwise}
                                          \end{cases}$
\end{itemize}

\textbf{2. Substitution of instance variables $p[l.@v := e]$}

The substitution for instance variables needs to take aliasing into account. For this, it is handy to have
conditionals in the assertion language.

\noindent It is defined by induction on the structure of $p$:
\begin{itemize}
 \item $l[l.@v := e] \equiv l$ for $l \equiv x, v, \self, n, true, false$ (constants - unaffected)
 \item $l'.@v'[l.@v := e] \equiv \begin{cases}
                                 \PL{\textbf{if} } l' = l \PL{ \textbf{then} } e \PL{ \textbf{else} } l'.@v' \PL{ \textbf{end}}\text{ if $@v = @v'$} \\
                                 l'.@v' \text{ otherwise} \\
                                \end{cases}$
 \item $\PL{\textbf{if} } l' \PL{ \textbf{then} } l_1 \PL{ \textbf{else} } l_2 \PL{ \textbf{end}}[l.@v := e] \equiv \PL{\textbf{if} } l'[l.@v := e] \PL{ \textbf{then} } l_1[l.@v := e] \PL{ \textbf{else} } l_2[l.@v := e] \PL{ \textbf{end}}$
 \item $l_1 = l_2 \equiv l_1[l.@v := e] = l_2[l.@v := e]$
 \item $l_1 < l_2 \equiv l_1[l.@v := e] < l_2[l.@v := e]$
 \item $\llbracket l' \rrbracket \in \{C_1,...,C_n\} \equiv \llbracket l'[l.@v := e] \rrbracket \in \{C_1,...,C_n\}$

 \item $(l_1 \wedge l_2)[l.@v := e] \equiv l_1[l.@v := e] \wedge l_2[l.@v := e]$
 \item $(l_1 \vee l_2)[l.@v := e] \equiv l_1[l.@v := e] \vee l_2[l.@v := e]$
 \item $(\neg l_1)[l.@v := e]\equiv \neg l_1[l.@v := e]$
 \item $(\exists y:T. l_1)[l.@v := e] \equiv \exists y:T. l_1[l.@v := e]$
 \item $(\forall y:T. l_1)[l.@v := e] \equiv \forall y:T. l_1[l.@v := e]$
\end{itemize}

\begin{lemma}[Substitution of instance variables]\todo{do we need this Lemma?}
\label{lem:sub-instvar}
For all logical expressions $s$ and $t$, all assertions $p$, all instance variables $@u$ and all
proper states $\sigma$
\begin{align}
\sigma(s[@u := t]) = \sigma[@u := t](s) \\
\sigma \models p[@u := t] \text{ iff } \sigma[@u := \sigma(t)] \models p.
\end{align}
\end{lemma}

\begin{proof}
 By induction on the structure of $s$ and $p$. $\Box$
\end{proof}

\textbf{2. Substitution for object creation $p[x := \mathbf{new}_C]$}

The substitution for object creation calculates the weakest precondition of an object creation statement.
For a slightly simpler case without classes, \cite[page 221]{AptDeBoerOlderogBook2009} defines a substitution
$[x := \mathbf{new}]$.
This substitution, however, is only applicable to so-called ``pure'' assertions. Fortunately, except for
conditionals, our logical expressions satisfy all requirements and \cite[page 223]{AptDeBoerOlderogBook2009}
gives a Lemma that allows eliminating conditionals like ours by substituting them for logically equivalent
expressions. We can thus use the substitution and only need to modify it slightly to reflect the
addition of classes.

The substitution is then defined by induction on the structure of $p$:
\begin{itemize}
 \item $l[x := \mathbf{new}_C] = l$ for $l \equiv \self, null, v, x, n, true, false$
 \item $l.@v[x := \mathbf{new}_C] \equiv \begin{cases} \mathit{init}_C.@v \text{ if $l \equiv x$} \\
                                                       l.@v \text{ otherwise}
                                         \end{cases}$
 \item $(\llbracket x \rrbracket \in \{C_1,...,C_n\})[x := \mathbf{new}_C] \equiv C \in \{C_1,...,C_n\}$

 \item $(l_1 = l_2)[x := \mathbf{new}_C] \equiv l_1[x := \mathbf{new}_C] = l_2[x := \mathbf{new}_C]$ if $l_1 \not\equiv x$ and $l_1 \not\equiv if...end$ and $l_2 \not= x$ and $l_2 \not\equiv if...end$
 \item $(x = l_2)[x := \mathbf{new}_C] \equiv false$ if $l_2 \not\equiv x$ and $l_2 \not\equiv if...end$ (also the symmetric case)
 \item $(x = x)[x := \mathbf{new}_C] \equiv true$

 \item $(\PL{\textbf{if} } l_0 \PL{ \textbf{then} } l_1 \PL{ \textbf{else} } l_2 \PL{ \textbf{end}} = l')[x := \mathbf{new}_C] \equiv \PL{\textbf{if} } l_0[x := \mathbf{new}_C] \PL{ \textbf{then} } (l_1 = l')[x := \mathbf{new}_C] \PL{ \textbf{else} } (l_2 = l')[x := \mathbf{new}_C] \PL{ \textbf{end}}$
 
 \item $\PL{\textbf{if} } l' \PL{ \textbf{then} } l_1 \PL{ \textbf{else} } l_2 \PL{ \textbf{fi}}[x := \mathbf{new}_C] \equiv \PL{\textbf{if} } l'[x := \mathbf{new}_C] \PL{ \textbf{then} } l_1[x := \mathbf{new}_C] \PL{ \textbf{else} } l_2[x := \mathbf{new}_C] \PL{ \textbf{fi}}$
 
 Note: conditionals can always be moved outwards to be the outmost operation in an assertion.
 
 \item $l_1 < l_2 \equiv l_1[x := \mathbf{new}_C] < l_2[x := \mathbf{new}_C]$

 \item $(l_1 \wedge l_2)[x := \mathbf{new}_C] \equiv l_1[x := \mathbf{new}_C] \wedge l_2[x := \mathbf{new}_C]$
 \item $(l_1 \vee l_2)[x := \mathbf{new}_C] \equiv l_1[x := \mathbf{new}_C] \vee l_2[x := \mathbf{new}_C]$
 \item $(\neg l_1)[x := \mathbf{new}_C]\equiv \neg l_1[x := \mathbf{new}_C]$
 \item $(\exists y:T. l_1)[x := \mathbf{new}_C] \equiv \exists y:T. l_1[x := \mathbf{new}_C] \vee l_1[x/y][x := \mathbf{new}_C]$
 \item $(\forall y:T. l_1)[x := \mathbf{new}_C] \equiv \forall y:T. l_1[x := \mathbf{new}_C] \wedge l_1[x/y][x := \mathbf{new}_C]$
\end{itemize}

\section{Clarke's Incompleteness Result \& Turing Completeness}
\label{app:turing_complete}
Clarke's Incompleteness Result \cite{Clarke1979} demonstrates that there are programming languages for which
no sound and complete Hoare Logic can exist. Since no sound and complete Hoare Logic was proposed for a
dynamically-typed programming language before, it is interesting to study whether this is at all possible.

However, the argument of Clarke is not applicable to \langDyn{} for three reasons:
\begin{enumerate}
 \item \langDyn{} does not satisfy the assumption that the expressions used in the programming language are a subset of those used in the assertion language.
This is only the case for statically-typed languages.
 \item \langDyn{} does not fulfill the language requirements Clarke bases his argument on.
In particular, it features neither
global variables nor internal procedures nor does it allow passing procedure names as parameters.
 \item Indeed, \langDyn{} ceases to be Turing complete under a finite interpretation\footnote{%
Clarke's programming language uses Integer variables.
In this case ``finite interpretation'' means restricting the program variables
to the finite subset $\{0,1\}$ of their original domain.
Since in \langDyn{} variables are of object type, interpreting them
finitely means bounding the number of objects on the Heap to a finite number.}
which will be explained in the following section.
\end{enumerate}

\subsection{Turing completeness}

\langDyn{} is of course Turing complete. Writing a \langDyn{} program simulating a Turing machine
is a straightforward excercise. However, it is not that easy to see that this expressiveness stems
only from the fact that \langDyn{} programs are allowed to create an unbounded number of objects.
In particular, while the stack depth in \langDyn{} is also unbounded, it is only possible to access
a finite number of variables at the top of the stack (the local variables of the current method) without
pop'ing (exiting the current method) which only yields the expressive power of push-down automata rather
than that of queue- or Turing-machines.
To see that this is the case, consider the following construction:
\begin{enumerate}
 \item bounding the number of objects on the Heap to some limit $k \in \mathbb{N}$ can be achieved by
 introducing a global counter and letting object creation fail once the limit is reached.
 \item It is straightforward to rewrite \langDyn's operational semantics (given in Section~\ref{sec:dyn:operational})
 in such a way that it uses a stack to handle method calls instead of begin-local-blocks.
 \item Now the states can be separated into Stack and Heap. By identifying states with the same Heap, the
labelled transition system defined by the operational semantics becomes a directed graph. Note that there
are only finitely many possible Heaps containing $\le k$ objects. This implies that diverging programs must
have state cycles.
Annotating the edges not only by computation steps, but also by stack frames being pushed and popped on method
calls / returns is possible since we only have a finite number of objects and hence also a finite number of
possible stack frames to be pushed.
This way, we can for every \langDyn{} program obtain a push-down automaton that accepts all finite computations
of the original program with an empty stack.
 \item Since emptyness is decidable for push-down automata, we hence have a method to check whether or not a given
 \langDyn{} program has a finite computation. If is does not, it will surely not halt. Thus the halting problem
for \langDyn{} with bounded Heap is decidable and \langDyn{} hence ceases to be Turing-complete when bounding the Heap.
\end{enumerate}

%\section{Variables}
%\label{app:variables}
%\input{sections/variables}
\fi

\end{document}